\documentclass[letter, a4paper]{IEEEtran}
\usepackage{amsmath,amsfonts}
\usepackage{algorithm}
\usepackage{array}
\usepackage{textcomp}
\usepackage{stfloats}
\usepackage{url}
\usepackage{verbatim}
\usepackage{graphicx}
\usepackage{cite}
\usepackage{float} 
\usepackage{authblk}
\usepackage{lipsum}
\usepackage[mathscr]{euscript}
\usepackage{framed}
\usepackage{algpseudocode}
\usepackage{bm}
\usepackage{stfloats}  
\usepackage{setspace}
\usepackage{balance}
\usepackage{amsthm} 

\usepackage{etoolbox}

\usepackage[caption=false,font=footnotesize]{subfig} 

\usepackage{graphicx,amsmath,amssymb,amsfonts}
\allowdisplaybreaks[3]
\usepackage{tikz}
\usetikzlibrary{arrows.meta,backgrounds,calc,decorations.pathreplacing,fit,positioning,shapes.geometric}
\usepackage{hyperref}
\hypersetup{colorlinks=true}   
\usepackage{xcolor}
\usepackage[T1]{fontenc}   
\newtheorem{remark}{\textbf{Remark}}
\newtheorem{lemma}{\textbf{Lemma}}

\newtheorem{corollary}{\textbf{Corollary}}

\floatname{algorithm}{Algorithm}

\makeatletter
\renewcommand{\maketag@@@}[1]{\hbox{\m@th\normalsize\normalfont#1}}%
\makeatother

\makeatletter

\newcommand{\Rmnum}[1]{\expandafter\@slowromancap\romannumeral #1@}

\makeatother

\begin{document}
	\captionsetup{font={small}}
	\bstctlcite{reference:BSTcontrol}
	\title{\fontsize{22 pt}{\baselineskip}\selectfont  Sparse Rotatable Array (SRA): Unifying Array Aperture and Antenna Directivity for Wireless Communications}
	\author{
		\fontsize{10 pt}{\baselineskip}\selectfont Ailing~Zheng, Qingqing~Wu, Xiyuan~Liu, Wen~Chen
		\vspace{-11 mm}
		\thanks{A. Zheng, Q. Wu, X. Liu, and W. Chen are with the School of Integrated Circuit, Shanghai Jiao Tong University, Shanghai 200240, China (e-mail: {ailing.zheng, qingqingwu, xiyuanliu\_SJTU, wenchen}@sjtu.edu.cn).}
		
	}
	
	
	\maketitle
	\begin{abstract} 
		Sparse rotatable array (SRA) is a novel reconfigurable antenna architecture that jointly exploits sparse aperture configuration and antenna directivity to enhance spatial resolution for future wireless communications. Specifically, SRA activates a subset of rotatable antennas over a large candidate aperture and adjusts their boresight directions, thereby creating a directionally selective sparse aperture with reduced hardware requirements and enhanced spatial flexibility.
In this paper, we investigate an SRA-aided multi-group communication system, where users are organized into spatial groups with different service requirements. We develop a group-aware SRA design framework by jointly optimizing the sparse-aperture allocation, RA orientations, and transmit beamforming to maximize the weighted max-min signal-to-interference-plus-noise ratio (SINR).
	Then, we characterize the operating principles of SRA and reveal that sparse aperture improves spatial resolution by enlarging the effective array aperture, while antenna directivity suppresses inter-group coupling through directional control, thereby enabling simplified group-wise beamforming structures.
		Guided by these insights, we develop a structured low-complexity  alternating optimization algorithm that embeds a closed-form projected group-center RA orientation rule into the sparse-aperture allocation and beamforming design. The proposed algorithm combines analysis-guided initialization, sampled multi-start antenna-allocation search, and bisection-based second-order cone programming for beamforming. Numerical results show that the proposed SRA design closely approaches fully-shared SRA benchmarks and significantly outperforms compact subarray and omni sparse-array schemes.
	\end{abstract}
	
		\vspace{-2mm}
\begin{IEEEkeywords}
	sparse rotatable array (SRA), multiuser communications, spatial resolution, sparse aperture allocation, antenna directivity, interference suppression.
\end{IEEEkeywords}

	\vspace{-5mm}
	\section{Introduction}
	\vspace{-2mm}
Future wireless networks are expected to support unprecedented data rates, massive connectivity, and diverse emerging applications while meeting increasingly stringent requirements on energy efficiency and deployment cost \cite{6GSaad}. These demands call for more efficient exploitation of the spatial domain, which has become a fundamental resource for supporting high-capacity multiuser transmission and interference management. 
Extremely large-scale multiple-input multiple-output (XL-MIMO) has emerged as a promising technology for sixth-generation (6G) wireless communications by enlarging physical apertures and providing abundant spatial degrees of freedom (DoFs) \cite{Lu2024MIMO}. Its large aperture enables substantial array gains and fine spatial resolution, thereby improving spatial multiplexing capability and supporting efficient multiuser interference management.
However, conventional XL-MIMO architectures mainly achieve such spatial enhancement through dense deployment of antennas over large physical apertures. As the array size increases, this approach entails rapidly growing hardware cost, energy consumption, and computational burden. Moreover, the spatial transmission characteristics of each antenna remain largely constrained by fixed antenna locations and radiation patterns, restricting the flexibility of spatial resource configuration. 
Therefore, innovative research on scalable spatial transmission architectures that can flexibly configure large-aperture spatial resources through adaptive aperture and radiation control while reducing hardware requirements  remains imperative for future wireless networks.

In this paper, we introduce a sparse rotatable array (SRA) architecture as a promising new solution to achieve the above goal. Specifically, SRA activates only a subset of candidate rotatable antennas (RAs) distributed over a large physical aperture, while each activated RA can independently adjust its boresight direction. Sparse activation substantially reduces the required radio-frequency (RF) chains while preserving a large effective aperture for enhanced spatial resolution. Meanwhile, boresight reconfiguration provides flexible directional control through adjustable RA orientations.  
Therefore, SRA jointly configures the effective array aperture and antenna directivity through RA activation and boresight adjustment, thereby integrating aperture and directivity reconfiguration within a unified array architecture.
The proposed SRA architecture has recently been
investigated in wireless sensing applications \cite{Ye2026}, where
SRA is exploited to alleviate spatial ambiguity
and improve sensing performance. However, its potential for
multiuser communications remains largely unexplored.

SRA enables large sparse apertures to achieve enhanced spatial resolution beyond aperture enlargement alone.
Specifically, although increasing the array aperture improves spatial resolution, conventional sparse arrays face inherent challenges in controlling grating lobes. Periodic sparse layouts may generate strong grating lobes due to spatial aliasing, while non-uniform sparse layouts can alleviate such periodic ambiguities but generally retain residual spatial leakage \cite{Chen2026,Chen2024,Zhou2025}.
By adjusting the boresight directions of activated RA toward intended spatial regions,  SRA concentrates the available directional gain where useful transmission is desired and attenuates radiation toward undesired grating-lobe regions. The resulting directionally selective sparse aperture therefore preserves the resolution benefit of a large physical aperture while simultaneously enhancing useful radiation and limiting spatial leakage.
The above capability originates from a distinct spatial reconfiguration principle that differentiates SRA from existing flexible antenna architectures. Movable antennas (MAs) and fluid antennas (FAs) adapt the propagation geometry by relocating antennas within a prescribed region \cite{Wong2020,Zhu2023,Ma2023MIMO,Zhu2023MA}, typically requiring additional movement space or position-control mechanisms. RAs, in contrast, retain fixed antenna locations and adjust their boresight directions to reshape element-level radiation gains, with existing studies mainly considering orientation adjustments over prescribed array layouts \cite{ZhengRA2025,AilngRASecure2026,AilngRA2026,Zheng2026ISAC}. SRA instead configures the activated RAs over a large candidate aperture to jointly configure the effective aperture and antenna orientations, thereby enabling flexible spatial resource reconfiguration.

The joint configurability of the effective aperture and antenna directivity in SRA enables a new paradigm for flexible multiuser transmission. In practical wireless networks, users are often distributed non-uniformly in space, where different service regions may exhibit distinct user densities, spatial separations, and channel characteristics. Such spatial heterogeneity leads to different requirements on spatial resolution and interference suppression. A fully-shared SRA offers the greatest transmission flexibility by allowing all activated RAs to jointly serve all users. However, it requires joint beamforming and RA orientation design over the entire sparse aperture, resulting in a high-dimensional optimization problem whose complexity increases rapidly with the numbers of activated RAs and users. 
Moreover, since the directional gain of each RA is concentrated around its boresight, an RA oriented toward one service region provides only limited useful gain to other well-separated regions. Therefore, full-array cooperation may offer only limited additional benefits when the service regions are sufficiently separated.
Motivated by these considerations, we consider a group-aware SRA architecture in which users with similar spatial characteristics are clustered into service groups, as shown in Fig.~\ref{Fig1}. The activated RAs are partitioned into group-specific sparse subarrays, allowing each group to configure its effective aperture according to its spatial demand. This structured partition transforms the globally coupled design into group-wise sparse-aperture allocation and reduced-dimensional beamforming. Consequently, each group can leverage a tailored effective aperture to enhance intra-group spatial separability while utilizing RA directivity to suppress inter-group leakage.

However, designing an efficient configuration strategy for such a group-aware SRA architecture remains challenging. The SRA configuration is inherently combinatorial, since the number of activated RAs assigned to each group and their sparse positions over the candidate aperture must be jointly optimized. Moreover, these configuration variables are tightly coupled with the resulting effective channel structure, including array gain, intra-group channel correlation, directional leakage, and inter-group interference. Since each candidate configuration requires a corresponding beamforming optimization, exhaustive search over possible SRA configurations becomes computationally prohibitive. Therefore, an efficient design framework is needed to fully exploit the spatial adaptability of SRA while managing the coupled effects of aperture configuration and antenna directivity. To address these challenges, this paper investigates an SRA-enabled multi-group communication system. The main contributions are summarized as follows.

			\begin{figure}[t]
	\centering
	\includegraphics[width=0.4 \textwidth]{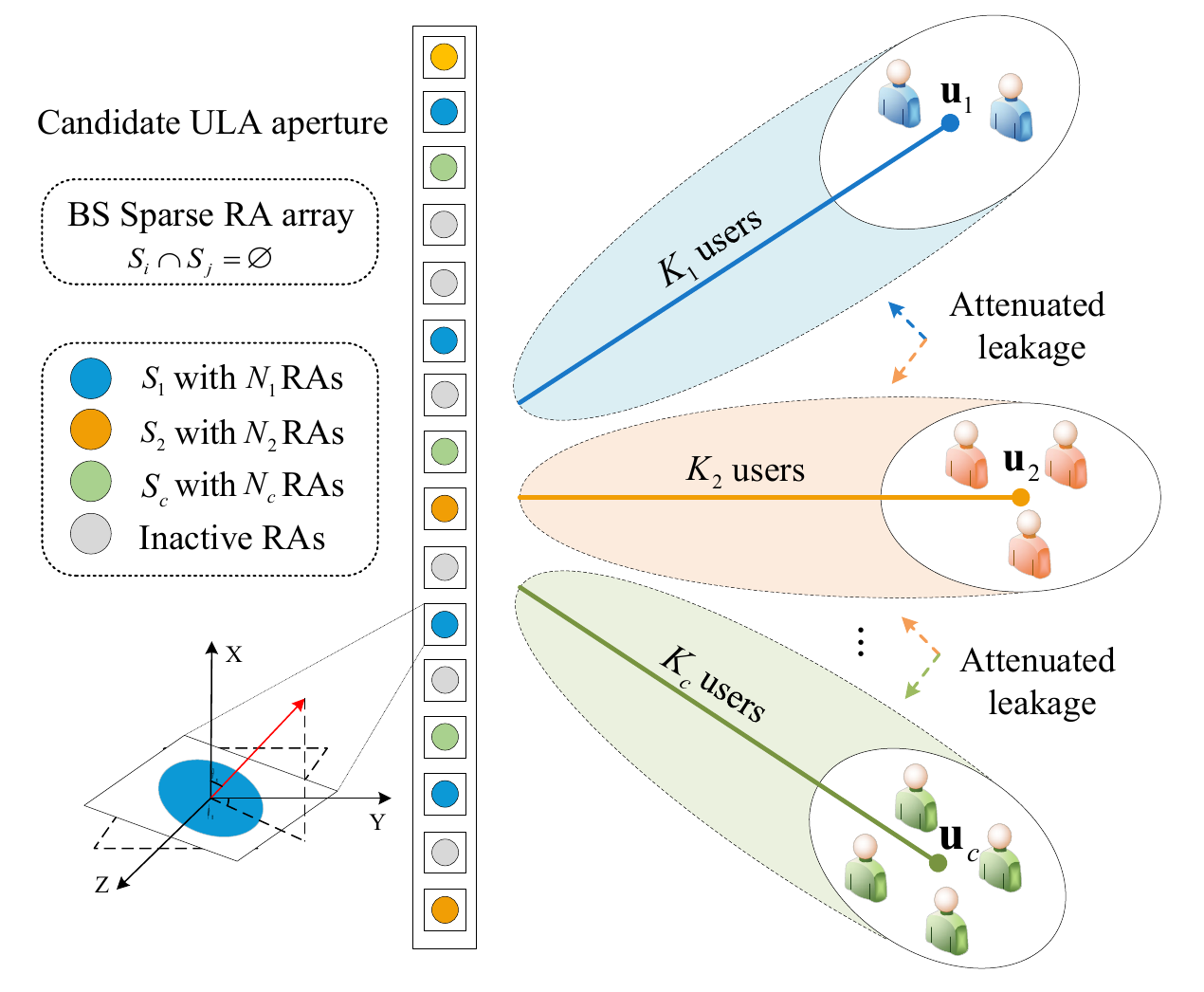}
	\caption{Illustration of the proposed SRA architecture with group-specific sparse subarrays.}
	\label{Fig1}
	\vspace{-7mm}
\end{figure}
\begin{itemize}
	
	\item 
	We propose an SRA architecture for multiuser communications and develop a group-aware SRA framework to accommodate heterogeneous spatial service requirements, where the activated RAs are partitioned into group-specific sparse subarrays. Based on this framework, a weighted max-min signal-to-interference-plus-noise ratio (SINR) optimization problem is formulated by jointly optimizing the sparse-aperture allocation, RA orientations, and transmit beamforming.

	\item
	We characterize the operating principles of SRA by revealing
	the impacts of sparse-aperture configuration and antenna directivity on multiuser channel structure. Specifically, sparse
	apertures improve intra-group spatial separability, leading to
	alignment between group-wise zero-forcing (ZF) and maximum-ratio transmission (MRT) beam directions, while RA
	directivity suppresses inter-group coupling, enabling global ZF to
	approach block-diagonal group-wise ZF. These insights reveal a
	decoupled group-wise beamforming structure under low intra-group correlation and weak inter-group coupling, and provide physical guidance for group-wise antenna-number allocation and non-periodic sparse-position initialization.

	\item
		We develop a structured low-complexity alternating optimization (AO) framework for SRA design. Starting from the analysis-guided antenna-number and non-periodic position initialization, the proposed algorithm refines the sparse-aperture configuration through a sampled multi-start greedy search. Specifically, for each candidate allocation, the corresponding RA orientations are directly determined by a closed-form projected group-center rule, thereby avoiding iterative orientation optimization, and the resulting reduced-dimensional weighted max-min beamforming problem is solved via bisection and second-order cone programming (SOCP).

	\item
	Numerical results demonstrate the effectiveness of the proposed SRA design compared with various benchmark schemes. The results verify that jointly configuring the sparse aperture and antenna directivity can effectively improve multiuser transmission performance by enhancing intra-group spatial separability and suppressing inter-group leakage. In particular, the proposed SRA design achieves approximately $12.2$-dB and $5.4$-dB gains in weighted max-min SINR over the omni sparse-array scheme and the compact subarray scheme, respectively.

\end{itemize}

	The remainder of this paper is organized as follows. Section \ref{System Model}  introduces the system model and problem formulation. Section \ref{Performance Analysis} characterizes the spatial operating principles of SRA and derives the resulting simplified beamforming structures under low intra-group correlation and weak inter-group coupling. Section \ref{Proposed Solution} develops the proposed low-complexity AO framework. Section \ref{Simulation Results} provides numerical results and Section \ref{Conclusion} concludes this paper.

\vspace{-3mm}
		\section{System Model}  \label{System Model} 
\vspace{-1mm}

	\begin{figure}[t]
		\centering
		\includegraphics[width=0.47 \textwidth]{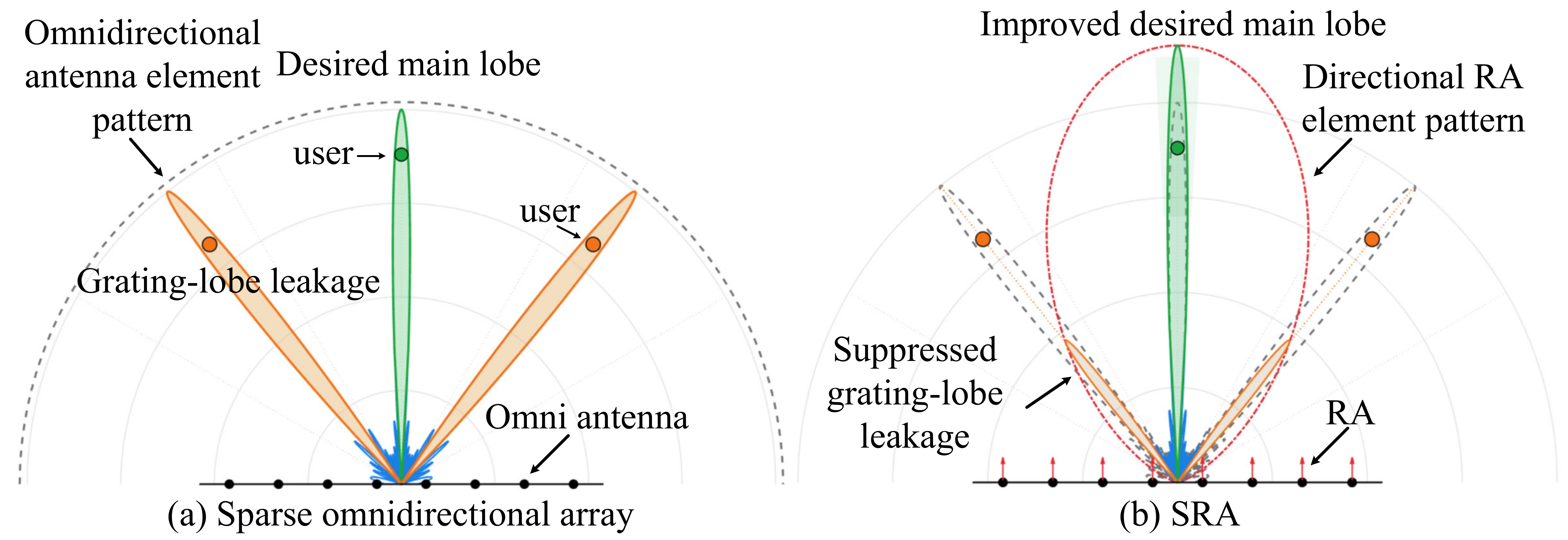}
		\caption{Directional-gain enhancement and grating-lobe suppression of the SRA.
			(a) A sparse array with omnidirectional elements suffers from strong grating-lobe leakage.
			(b) The directional RA element pattern enhances the desired main lobe while suppressing off-boresight grating-lobe leakage.}
			\label{fig:RA_GL_suppression}
		\vspace{-7mm}
	\end{figure}
	We consider a downlink SRA-assisted multi-group communication
	system, where a base station (BS) serves multiple single-antenna users distributed
	in a three-dimensional coverage region $\mathcal A$. 
The grating-lobe suppression mechanism of SRA enabled by sparse aperture and RA directivity is illustrated in Fig.~\ref{fig:RA_GL_suppression}.
	 Based on the spatial user distribution, the users in $\mathcal A$ are partitioned into multiple service groups.
	Let $\mathcal C=\{1,\ldots,C\}$ denote the set of service groups.
	For each service group $c\in\mathcal C$, $\mathcal K_c$ denotes the set
	of users in this group, and $K_c=|\mathcal K_c|\ge 1$ denotes the
	corresponding number of users. The vector
	$\mathbf u_c=[x_c,y_c,z_c]^T$ denotes the representative center of
	service group $c$. A service group may correspond to a compact hotspot
	group or a singleton group. When $K_c=1$, the group center
	reduces to the location of the only user, i.e.,
	$\mathbf u_c=\mathbf u_{c,1}$. The total number of users is
	$K=\sum_{c\in\mathcal C}K_c$. For user $k\in\mathcal K_c$, its local coordinate relative to the
	representative center of service group $c$ is denoted by
	$\mathbf u_{c,k}^{l}=[x_{c,k}^{l},y_{c,k}^{l},z_{c,k}^{l}]^T$.
	Thus, the global coordinate of user $k\in\mathcal K_c$ is given by
	$\mathbf u_{c,k}=\mathbf u_c+\mathbf u_{c,k}^{l}$. For a single-user
	service group, $\mathbf u_{c,1}^{l}=\mathbf 0$.
	The BS is located at $\mathbf t_0=[x_0,y_0,z_0]^T$ and is equipped with a candidate uniform linear array (ULA) aperture consisting of $M$ potential RA positions, among which $N\leq M$ RAs are activated~\cite{Chen2024}.
	The local coordinate of the $m$-th candidate
	RA position is denoted by
	$\mathbf t_m^l=[x_m^l,0,0]^T$, $m=1,\ldots,M$, where
	$x_m^l=\left(m-\frac{M+1}{2}\right)\Delta$ and $\Delta$ denotes the
	inter-grid spacing. Accordingly, the global coordinate of the $m$-th
	candidate RA position is given by
	$\mathbf t_m=\mathbf t_0+\mathbf t_m^l$.

	Different from conventional fully shared-array transmission, the considered SRA system activates only a subset of candidate RAs over a large physical aperture. For multi-group transmission, the activated RAs are partitioned into $C$ disjoint group-specific sparse subarrays, where each subarray serves one service group. Let $N_c$ denote the number of RAs assigned to service group $c$, with $N=\sum_{c\in\mathcal C}N_c$. The activated RAs assigned to each service group form a group-specific sparse aperture, where the sparse aperture configuration and boresight directions jointly determine the effective channel characteristics.
	To characterize the sparse aperture allocation, we define the binary
	antenna allocation variable $b_{m,c}\in\{0,1\}$, where $b_{m,c}=1$
	indicates that candidate position $m$ is assigned to service group $c$,
	and $b_{m,c}=0$ otherwise. Thus, the index set of the subarray
	serving group $c$ is given by
	$\mathcal S_c=\{m\mid b_{m,c}=1\}$. The antenna allocation satisfies
	\setlength\abovedisplayskip{1pt}
	\setlength\belowdisplayskip{1pt}
	\begin{align}
		&\sum\nolimits_{c\in\mathcal C} b_{m,c}\le 1, \forall m,
		N_c=\sum\nolimits_{m=1}^{M} b_{m,c}, \forall c\in\mathcal C,\\
		&N_c\ge N_c^{\min}, \forall c\in\mathcal C,
		\sum\nolimits_{c\in\mathcal C}N_c=N.
	\end{align}
	
	Let $\mathbf f_m$ denote the boresight direction of the $m$-th RA, which is parameterized by the zenith angle $\phi_m$ and azimuth angle $\theta_m$ as
$\mathbf f_m
		=
		[\sin\phi_m\cos\theta_m,\sin\phi_m\sin\theta_m,\cos\phi_m]^T$.
	The practical zenith-angle constraint is given by $0\le \phi_m\le \phi_{\max}$, where $\phi_{\max}\in[0,\pi/2]$ denotes the maximum adjustable zenith angle. Equivalently, the boresight vector satisfies
	\begin{equation}
		\label{angel_con}
		\cos(\phi_{\max})\le \mathbf f_m^T\mathbf e_z\le 1,
	\end{equation}
	where $\mathbf e_z=[0,0,1]^T$ is the unit vector along the $z$-axis. 
	Boresight directions are specified only for activated positions, since inactive positions do not contribute to transmission.
		The effective antenna gain for each RA depends on the signal arrival/departure angle and antenna directional gain pattern. The directional gain pattern of each RA is modeled as \cite{ZhengRA2026}
	\begin{equation}
		G(\epsilon)=
		\begin{cases}
			G_0\cos^{2p}(\epsilon),
			&
			\epsilon\in\left[0,\frac{\pi}{2}\right],
			\\
			0,
			&
			\text{otherwise},
		\end{cases}
	\end{equation}
	where $\epsilon$ denotes the angular offset between the signal direction and the RA boresight,
	$p\ge0$ controls the antenna directivity,
	$G_0=2(2p+1)$ denotes the maximum directional gain.
	
	\vspace{-4mm}
	\subsection{Channel Model}
	\vspace{-1mm}
	We adopt a narrow-band geometric frequency-flat channel model. 
	The perfect channel state information (CSI) is assumed to reveal the fundamental performance of the proposed SRA-assisted system. 
	The line-of-sight (LoS) link power gain from transmit RA $m$ to user $k$ in group $c$ is given by \cite{Peng2026}
	\begin{align}
		G_{m,c,k}^{\mathrm{LoS}}(\mathbf{f}_{m}) 	&=\!\beta_0r_{m,c,k}^{-2}G(\epsilon_{m,c,k}) \nonumber \\ 
		&=\!\beta_0r_{m,c,k}^{-2}G_0\left[\!\frac{\mathbf{f}_{m}^\mathrm{T}(\mathbf{u}_{c,k} - {\mathbf{t}}_{m})}{r_{m,c,k}}\!\right]_{+}^{2p},
	\end{align}
	where $\beta_0 = (\frac{\lambda}{4\pi })^2$ is the free-space reference gain constant, $\lambda$ denotes the wavelength, $r_{m,c,k} = \Vert \mathbf{u}_{c,k}-\mathbf{t}_m\Vert_2$ is the distance between the transmit RA $m$ and the user $\mathbf{u}_{c,k}$. Here, $\cos(\epsilon_{m,c,k}) =\frac{(\mathbf{f}_{m})^\mathrm{T}(\mathbf{u}_{c,k} - {\mathbf{t}}_{m})}{r_{m,c,k}}$ represents the cosine of the angle between the transmit RA $m$'s orientation $\mathbf{f}_{m}$ and the LoS direction to the user $\mathbf{u}_{c,k}$. The related channel is expressed as
	\begin{align}
		\!\!\!\!\!\!\! h_{m,c,k}^{\mathrm{LoS}}(\mathbf{f}_{m})&\!=\! \sqrt{G_{m,c,k}^{\mathrm{LoS}}(\mathbf{f}_{m})}e^{-j\frac{2\pi}{\lambda}r_{m,c,k}} \nonumber \\
		&\! =\!\frac{\sqrt{\beta_0G_0}}{r_{m,c,k}} \left[\frac{\mathbf{f}_{m}^\mathrm{T}(\mathbf{u}_{c,k}-{\mathbf{t}}_{m})}{r_{m,c,k}}\right]_{+}^{p}e^{-j\frac{2\pi}{\lambda}r_{m,c,k}}.
	\end{align}
	The channel coefficient captures both the directional antenna gain and the propagation-induced phase shift.
	
	We further consider a scattering environment, where the group $c$ is
	associated with $D_c$ spatially distributed scatterers. The location of
	scatterer $d$ associated with group $c$ is denoted by
	$\mathbf s_{c,d}\in\mathbb R^3$. The non-LoS (NLoS) power gain from transmit RA
	$m$ to scatterer $d$ is given by
	\begin{align}
		G^{\rm NLoS}_{m,c,d}(\mathbf f_m)
		&=
		\beta_0 r_{m,c,d}^{-2}G(\epsilon_{m,c,d}) \nonumber\\
		&=
		\beta_0 r_{m,c,d}^{-2}G_0
		\left[
		\frac{\mathbf f_m^T(\mathbf s_{c,d}-\mathbf t_m)}
		{r_{m,c,d}}
		\right]_+^{2p},
	\end{align}
	where $r_{m,c,d}
	=\|\mathbf s_{c,d}-\mathbf t_m\|_2$
	is the antenna-to-scatterer distance, and $\cos(\epsilon_{m,c,d})
	=\frac{\mathbf f_m^T(\mathbf s_{c,d}-\mathbf t_m)}
	{r_{m,c,d}}$.
	Considering a bi-static scattering model, the NLoS channel coefficient
	from transmit RA $m$ to user $k$ in group $c$ is expressed as
	\begin{align}
		h^{\rm NLoS}_{m,c,k}(\mathbf f_m)
		=
		\sum\nolimits_{d=1}^{D_c}
		a_{m,c,k,d}
		\left[
		\frac{\mathbf f_m^T(\mathbf s_{c,d}-\mathbf t_m)}
		{r_{m,c,d}}
		\right]_+^p,
	\end{align}
	where
	\begin{align}
		&a_{m,c,k,d}
		=
		\sqrt{\frac{\sigma_{c,d}\beta_0G_0}{4\pi}}\,
		\frac{
			e^{-j\frac{2\pi}{\lambda}(r_{m,c,d}+r_{c,d,k})+j\chi_{c,d}}
		}
		{r_{m,c,d}r_{c,d,k}},
	\end{align}
	and $r_{c,d,k}
	=\|\mathbf s_{c,d}-\mathbf u_{c,k}\|_2 $ is the scatterer-to-user distance. Here, $\sigma_{c,d}$ denotes the
	scattering coefficient of scatterer $d$ in
	group $c$, and $\chi_{c,d}\sim\mathcal U(0,2\pi)$ denotes the
	corresponding random scattering phase.
	Then, we have
	\begin{align}
		h_{m,c,k}(\mathbf{f}_{m}) = h_{m,c,k}^{\mathrm{LoS}}(\mathbf{f}_{m})+h_{m,c,k}^{\mathrm{NLoS}}(\mathbf{f}_{m}). 
	\end{align}
	The overall channel between the BS and the user $\mathbf{u}_{c,k}$ is
	\begin{align}
		\label{channel_h}
		&\!\!\! \mathbf{h}_{c,k} \!=\! [h_{1,c,k}(\mathbf{f}_{1}), h_{2,c,k}(\mathbf{f}_{2}), \ldots, h_{M,c,k}(\mathbf{f}_{M})]^T \!\in \!\mathbb C^{M\times 1}.  
	\end{align}

	\vspace{-4mm}
	\subsection{Transmission Model}
	\vspace{-1mm}
	Let $s_{c,k}$ denote the symbol intended for user $k$ in group $c$,
	where $\mathbb E[|s_{c,k}|^2]=1$. Since $N_c$ and $\mathcal S_c$ are
	determined by the optimization variables, we use an $M$-dimensional
	virtual beamforming vector $\mathbf w_{c,k}\in\mathbb C^{M\times 1}$
	for user $k$ in group $c$.
	Physically, group $c$ can only
	use the antennas assigned to $\mathcal S_c$. This is enforced by the
	support constraint $\mathbf B_c\mathbf w_{c,k}=\mathbf w_{c,k},
		\forall c, k\in\mathcal K_c$,
	where $\mathbf B = [\mathbf B_1, \mathbf B_2, \cdots, \mathbf B_C ]$ and $\mathbf B_c=\operatorname{diag}(b_{1,c},\ldots,b_{M,c})$.
	Thus, $w_{m,c,k}=0$ if $b_{m,c}=0$.
	The transmitted signal is
	\begin{equation}
		\mathbf x
		=
		\sum\nolimits_{c=1}^{C}\sum\nolimits_{k=1}^{ K_c}
		\mathbf w_{c,k}s_{c,k}.
	\end{equation}
	The received signal at user $\mathbf{u}_{c,k}$ is
	\begin{align}
		y_{c,k}
		\!\!=\!\!
		\mathbf h_{c,k}^H\mathbf w_{c,k} s_{c,k}
		\!\!	+\!\!\!\!\!\!\!\!
		\sum_{i \in \mathcal{K}_c, i\neq k} \!\!\!\!\!\!
		\mathbf h_{c,k}^H\mathbf w_{c,i} s_{c,i} \!\!+\!\!\!
		\sum_{j\neq c} \! \sum_{i=1}^{K_{j}} \!
		\mathbf h_{c,k}^H\mathbf w_{j,i} s_{j,i}
		\!+\!
		n_{c,k}, \nonumber 
	\end{align}
	where $	n_{c,k}\sim\mathcal{CN}(0,\sigma_k^2)$ denotes additive white Gaussian noise.
	The received SINR of user $k$ in group $c$ is expressed as
	\begin{align}
		\!\!\!\! \gamma_{c,k}
		\!=\!
		\frac{
			|\mathbf h_{c,k}^H\mathbf w_{c,k}|^2
		}{
			\sum_{i\neq k}
			|\mathbf h_{c,k}^H\mathbf w_{c,i}|^2 \!+\! \sum_{j\neq c}\sum_{i=1}^{K_{j}}
			|\mathbf h_{c,k}^H \mathbf w_{j,i}|^2
			\!+\! \sigma_k^2
		}. 
	\end{align}

	\vspace{-2mm}
	\subsection{Problem Formulation}
	\vspace{-1mm}
To account for different service priorities among groups, let $\Gamma_c>0$ denote the priority weight of group $c$.
	We maximize the minimum weighted SINR among all users by jointly optimizing the sparse aperture allocation, transmit beamforming, and RA orientations. The corresponding optimization problem is formulated as
	\begin{subequations}
		\label{P0}
		\begin{eqnarray}
			&\!\!\!\!\!\!\!\!\!\!\!\!\!\!\! \max \limits_{\mathbf B_{c},\mathbf w_{c,k}, \mathbf f_{m}}
			&\!\!\! \min_{c,k}
			\frac{\gamma_{c,k}}{\Gamma_c}
			\label{P1_obj}\\
			&\!\!\!\!\!\!\!\!\!\!\!\!\!\!\! \mathrm{s.t.}
			&\!\!\!\sum\nolimits_{c=1}^{C}\sum\nolimits_{k=1}^{K_c}
			\|\mathbf w_{c,k}\|_2^2
			\leq P_{\max},
			\label{P1_power}\\
			&&\!\!\! \mathbf B_c\mathbf w_{c,k}=\mathbf w_{c,k},
			\forall c,\; k\in\mathcal K_c,
			\label{P1_support}\\
			&&\!\!\! \sum\nolimits_{c=1}^{C}\sum\nolimits_{m=1}^{M}b_{m,c}=N,
			\label{P1_total_ant}\\
			&&\!\!\! \sum\nolimits_{c=1}^{C}b_{m,c}\leq 1,
			\forall m,
			\label{P1_exclusive}\\
			&&\!\!\! \sum\nolimits_{m=1}^{M}b_{m,c}\geq N_c^{\min}, b_{m,c}\in\{0,1\},
			\forall m,c,
			\label{P1_min_ant}\\
			&&\!\!\!	\cos(\phi_{\max})\le \mathbf f_m^T\mathbf e_z\le 1, \Vert \mathbf f_m \Vert_2 = 1, \forall m.
			\label{P1_fm} 
		\end{eqnarray} 
	\end{subequations} 
	Problem~\eqref{P0} is challenging due to the mixed-integer antenna allocation and its coupling with the RA orientations and transmit beamforming. The allocation variables jointly determine the group-wise antenna numbers, sparse positions, and effective channel structure. To address this problem, we develop a structured low-complexity AO algorithm.

	\vspace{-3mm}
\section{Spatial Operating Principles of SRA}
	\label{Performance Analysis}
	\vspace{-1mm}
In this section, we characterize the spatial operating principles of SRA by investigating the roles of its two complementary spatial mechanisms, i.e., the effective array aperture and antenna directivity, in shaping the multiuser channel structure. Specifically, we consider compact multiuser groups with $K_c\geq2$ and adopt the LoS channel model with $D_c=0$ to reveal the fundamental effects of SRA configuration. The sparse apertures formed by group-specific RA activation are analyzed in terms of intra-group channel orthogonality, while RA directivity is characterized for inter-group leakage suppression. These analyses then provide physical guidance for high-quality initialization of the proposed group-wise antenna-number allocation and non-periodic sparse-position initialization.

\vspace{-4mm}
\subsection{Effective-Array-Aperture-Induced Intra-Group Orthogonalization}
	\label{III-A}
	\vspace{-1mm}
	In this subsection, we characterize the role of the effective array aperture provided by sparse RA activation in shaping intra-group channel structure and the resulting group-wise beamforming design. We consider a compact multiuser service group \(c\) with \(K_c\ge 2\), the
	assigned antenna set \(\mathcal S_c\), and \(|\mathcal S_c|=N_c\). Let \(\tilde{\mathbf h}_{c,c,k}\in
	\mathbb C^{N_c\times 1}\) denote the reduced channel from the subarray
	assigned to group \(c\) to user \(k\in\mathcal K_c\). The normalized
	reduced channel is defined as
$
	\bar{\mathbf h}_{c,k}
	=
	{\tilde{\mathbf h}_{c,c,k}}/
	{\|\tilde{\mathbf h}_{c,c,k}\|_2}$.
	The maximum normalized intra-group channel correlation is defined as
$
	\varepsilon_c
	=
	\max_{k\ne i}
	\left|
	\bar{\mathbf h}_{c,k}^{H}\bar{\mathbf h}_{c,i}
	\right|$.
	A small \(\varepsilon_c\) indicates that the users in group \(c\) are well
	separated by the selected sparse aperture. 
	Let $	\bar{\mathbf H}_c
	=
	[\bar{\mathbf h}_{c,1},\ldots,\bar{\mathbf h}_{c,K_c}]^{H}
	\in \mathbb C^{K_c\times N_c}$.
	The normalized group-wise MRT direction for user $k$ is
	$\mathbf q_{c,k}^{\rm MRT}
	=
	\bar{\mathbf h}_{c,k}$.
	The corresponding group-wise ZF direction is
	$\mathbf q_{c,k}^{\rm ZF}
	=
	\bar{\mathbf H}_c^{H}
	\left(\bar{\mathbf H}_c\bar{\mathbf H}_c^{H}\right)^{-1}
	\mathbf e_k$,
	where $\mathbf e_k$ is the $k$-th canonical basis vector.
Its normalized beam direction is defined as
$\bar{\mathbf q}_{c,k}^{\rm ZF}
		=
		{\mathbf q_{c,k}^{\rm ZF}}/
		{\Vert\mathbf q_{c,k}^{\rm ZF}\Vert_2}$.
	
	\begin{lemma}[ZF-MRT beam-direction deviation under intra-group correlation]
		\label{pro1}
	\upshape	If $(K_c-1)\varepsilon_c<1$, then
	$\bar{\mathbf H}_c$ has full row rank. The phase-invariant
		distance between the normalized group-wise ZF and MRT
		beam directions satisfies
		\begin{align}
		\!\!\!	d_{c,k}^{\rm ZF-MRT}
			\!\triangleq \!\!\!
			\min_{\varphi\in[0,2\pi)}
			\left\|
			\bar{\mathbf q}_{c,k}^{\rm ZF}
		\!	\!-\!\!
			e^{j\varphi}\bar{\mathbf h}_{c,k}
			\right\|_2\!\! \leq\!\!
			\frac{
				\sqrt{2(K_c \! -\! 1)}\varepsilon_c
			}{
				\sqrt{1\!-\!(K_c \!-\!2)\varepsilon_c}
			}.
			\label{eq:zf_mrt_bound}
		\end{align}
		For fixed $K_c$,
		$d_{c,k}^{\rm ZF-MRT}
		=
		\mathcal O\!\left(
		\sqrt{K_c-1}\,\varepsilon_c
		\right)$
		as $\varepsilon_c\to0$.
			\end{lemma}
	\begin{proof}
		Please refer to Appendix \ref{AppendixA}.
	\end{proof}
		Lemma \ref{pro1} shows that the group-wise ZF beamformer approaches MRT as the intra-group channels become nearly orthogonal. To quantify this regime, we define the channel-orthogonality margin of group $c$ as
	$\zeta_c\triangleq\sqrt{K_c-1}\varepsilon_c$. Let $\eta_{\rm bf}\ll1$ and $\eta_{\rm int}\ll1$ denote the prescribed tolerances on the ZF-MRT beam-direction deviation and the normalized intra-group interference power under MRT, respectively. Since the latter is upper bounded by $(K_c-1)\varepsilon_c^2$, a conservative channel-correlation threshold is
	\begin{equation}
		\varepsilon_c^{\rm th}
		=
		\min\left\{
		\frac{\eta_{\rm bf}}{2\sqrt{K_c-1}},
		\sqrt{\frac{\eta_{\rm int}}{K_c-1}},
		\frac{1}{2(K_c-1)}
		\right\}.
	\end{equation}
	When $\varepsilon_c\leq\varepsilon_c^{\rm th}$, group-wise ZF and MRT have close beam directions, while the normalized MRT interference remains below $\eta_{\rm int}$. In particular, they become identical when $\varepsilon_c=0$.
Let $R_{c,k}=\lVert\mathbf u_{c,k}-\mathbf t_0\rVert_2$ denote the distance from the array reference point to user $k$ in group $c$. We define
\begin{equation}
	\sin\psi_{c,k}
	=
	\frac{\mathbf e_x^{T}(\mathbf u_{c,k}-\mathbf t_0)}
	{R_{c,k}},
	\Delta s_{k,i}
	=
	\sin\psi_{c,i}-\sin\psi_{c,k},
\end{equation}
where $\mathbf e_x=[1,0,0]^T$. We next analyze the intra-group channel correlation in the far-field and near-field regimes.

\textbf{1) Far-field regime:} In the far-field regime, the normalized intra-group correlation is approximated as
\begin{equation}
	\rho_{k,i}^{(c),\mathrm{FF}}
	\simeq
	\frac{1}{N_c}
	\sum_{m\in\mathcal S_c}
	\exp\left(
	-j\frac{2\pi}{\lambda}
	x_m^{\mathrm l}\Delta s_{k,i}
	\right).
\end{equation}
The actual correlation depends on the complete position set, while the effective aperture determines the available phase spread. 
Let $D_{c,x}^{\mathrm{eff}}\triangleq \max_{m\in\mathcal S_c}x_m^{l}-\min_{m\in\mathcal S_c}x_m^{l}$ denote the effective aperture length of the subarray assigned to group $c$ along the $x$-axis. 
Since exact orthogonality requires the phasors to span at least a semicircle, a necessary far-field condition is
$D_{c,x}^{\mathrm{eff}}
	\left|\Delta s_{k,i}\right|
	\geq
	\frac{\lambda}{2}$,
which yields the minimum aperture scale
\begin{equation}
	D_{k,i}^{\mathrm{orth,FF}}
	=
	\frac{\lambda}
	{2\left|\Delta s_{k,i}\right|}.
\end{equation}

\textbf{2) Near-field regime:} In the near regime, we define
\begin{equation}
	\Delta q_{k,i}
	=
	\frac{\cos^2\psi_{c,i}}{R_{c,i}}
	-
	\frac{\cos^2\psi_{c,k}}{R_{c,k}}.
\end{equation}
Under the Fresnel approximation, an aperture-only necessary condition for exact orthogonality is
\begin{equation}
	D_{k,i}^{\mathrm{orth,NF}}
	=
	\frac{\lambda}
	{
		\left|\Delta s_{k,i}\right|
		+
		\sqrt{
			\left|\Delta s_{k,i}\right|^2
			+
			\frac{\lambda\left|\Delta q_{k,i}\right|}{4}
		}
	}.
\end{equation}
These conditions determine whether exact orthogonality is geometrically possible.
In addition, to obtain a simple aperture scale for initialization, we adopt the conservative resolvability threshold $\chi_c=1/(\pi\varepsilon_c^{\rm th})$ by using the envelope bound $|\operatorname{sinc}(x)|\leq 1/(\pi|x|)$, which gives
\begin{equation}
	\max\left\{
		\frac{
			D_{c,x}^{\mathrm{eff}}
			\left|\Delta s_{k,i}\right|
		}{\lambda},
		\frac{
			\left(D_{c,x}^{\mathrm{eff}}\right)^2
			\left|\Delta q_{k,i}\right|
		}{\lambda}
		\right\}
	\geq
	\chi_c.
\end{equation}

Based on the above, the pairwise aperture scales are approximated as
\begin{align}
	D_{k,i}^{\mathrm{th,FF}}
	=
	\frac{\chi_c\lambda}
	{\left|\Delta s_{k,i}\right|},
	D_{k,i}^{\mathrm{th,NF}}
	=
	\min\left\{
		\frac{\chi_c\lambda}
		{\left|\Delta s_{k,i}\right|},
		\sqrt{
			\frac{\chi_c\lambda}
			{\left|\Delta q_{k,i}\right|}
		}
		\right\}. \nonumber 
\end{align}
The corresponding group-level aperture requirements are
\begin{equation}
	D_{c,x}^{\mathrm{th},\xi}
	=
	\max_{k\neq i}
	D_{k,i}^{\mathrm{th},\xi},
	\xi\in\{\mathrm{FF},\mathrm{NF}\},
\end{equation}
where a zero denominator indicates an inactive separability dimension and gives an infinite requirement. These results motivate assigning a sufficiently large effective aperture to each compact multiuser group and adopting non-periodic sparse positions to reduce intra-group channel correlation.

\vspace{-4mm}
\subsection{RA-Directivity-Induced Inter-Group Leakage Suppression and ZF Approximation}
	\label{III-B}
	\vspace{-1mm}
	
In this subsection, we characterize the role of antenna directivity in controlling undesired inter-group channel coupling. By aligning the boresight directions of activated RAs with their corresponding group centers, the antenna directivity suppresses the effective channel gains toward non-target groups, thereby reducing inter-group interference. We then derive far-field and near-field bounds on the off-diagonal components of the multi-group channel matrix to quantify the leakage suppression capability and establish the conditions for the global-to-group-wise ZF approximation.

	For group $c$, we define the channel matrix from the subarray assigned to
	group $j$ to the users in group $c$ as
	\begin{equation}
		\mathbf H_{c,j}
		=
		\begin{bmatrix}
			\tilde{\mathbf h}_{j, c,1}^{H},
			\tilde{\mathbf h}_{j, c,2}^{H},
			\ldots,
			\tilde{\mathbf h}_{j, c,K_c}^{H}
		\end{bmatrix}^T
		\in\mathbb C^{K_c\times N_j}.
	\end{equation}
	The diagonal block $\mathbf H_{c,c}$ is the desired intra-group
	channel of group $c$, whereas $\mathbf H_{c,j}$, $j\neq c$, represents
	the leakage channel from the subarray assigned to group $j$ to the
	users in group $c$.
	By stacking all groups, the global reduced channel matrix is $	\mathbf H_g$.
	Define its block-diagonal desired part as $\mathbf H_{\rm bd}
	=
	{\rm blkdiag}
	(\mathbf H_{1,1},\mathbf H_{2,2},\ldots,\mathbf H_{C,C})$,
	and the inter-group leakage part as $\mathbf E_{\rm inter}
	=
	\mathbf H_g-\mathbf H_{\rm bd}$.
	Thus, $\mathbf H_g$ is approximately block diagonal if
	$\mathbf E_{\rm inter}$ is sufficiently small. 
	We then define
	$s_{\rm bd}\triangleq\sigma_{\min}(\mathbf H_{\rm bd})$
	and
	$\kappa_{\rm bd}\triangleq
	\|\mathbf H_{\rm bd}\|_2/s_{\rm bd}$.

	\begin{lemma}[Global-to-group-wise ZF precoder deviation under inter-group coupling]
		\label{pro2}
		\upshape
		Suppose that each diagonal submatrix $\mathbf H_{c,c}$
		has full row rank and that
		$\|\mathbf E_{\rm inter}\|_2
		\leq
		\rho s_{\rm bd}$
		for $\rho\in(0,1)$. Then, the global ZF precoder is
		well defined and satisfies
		\begin{equation}
			\left\|
			\mathbf W_g^{\rm ZF}
			-
			\mathbf W_{\rm gw}^{\rm ZF}
			\right\|_2
			\leq
			\frac{
				1+2\kappa_{\rm bd}^2+\rho\kappa_{\rm bd}
			}{
				(1-\rho)^2s_{\rm bd}^2
			}
			\|\mathbf E_{\rm inter}\|_2,
			\label{eq:global_groupwise_zf_bound}
		\end{equation}
			where 
		\begin{align}
			&	\mathbf W_g^{\rm ZF}
			=
			\mathbf H_g^H
			(\mathbf H_g\mathbf H_g^H)^{-1}, \mathbf W_c^{\rm ZF}
			=
			\mathbf H_{c,c}^H
			(\mathbf H_{c,c}\mathbf H_{c,c}^H)^{-1}, \nonumber \\
			&	\mathbf W_{\rm gw}^{\rm ZF}
			=
			{\rm blkdiag}
			\left(
			\mathbf W_1^{\rm ZF},
			\mathbf W_2^{\rm ZF},
			\ldots,
			\mathbf W_C^{\rm ZF}
			\right). \nonumber
		\end{align}
	\end{lemma}
			\begin{proof}
		Please refer to Appendix~\ref{app2}.
	\end{proof}
	\vspace{-2mm}
	Lemma~\ref{pro2} shows that, under weak inter-group coupling, the deviation between the global ZF precoder and the block-diagonal collection of group-wise ZF precoders is bounded linearly in $\|\mathbf E_{\rm inter}\|_2$. We next characterize the attenuation of the off-diagonal components of the multi-group channel matrix induced by RA directionality. For both far-field and near-field propagation, each RA is steered toward the center of its assigned group, while the corresponding departure directions are treated differently in the two regimes.
	
	\textbf{1) Far-field regime:}
	In the far-field regime, the departure directions from all
	antennas in subarray $j$ toward a given user are approximately
	identical. Let
	$\mathbf d_c=\frac{\mathbf u_c-\mathbf t_0}{
	\|\mathbf u_c-\mathbf t_0\|_2}$
	and
	$\mathbf d_{c,k}=\frac{\mathbf u_{c,k}-\mathbf t_0}{
	\|\mathbf u_{c,k}-\mathbf t_0\|_2}$
	denote the reference-point directions toward the center of
	group $c$ and user $k\in\mathcal K_c$, respectively. The RA
	boresights in subarray $j$ satisfy
	$\mathbf f_{m,j}^\star=
	\frac{\mathbf u_{j}-\mathbf t_m}
	{\|\mathbf u_{j}-\mathbf t_m\|_2} \approx \mathbf d_{j}$.
We define the group-center angular separation and radius
	as
	\begin{equation}
		\Theta_{j,c}
		=
		\arccos
		\left(
		\mathbf d_j^T\mathbf d_c
		\right),
		\Delta_c
		=
		\max_{k\in\mathcal K_c}
		\arccos
		\left(
		\mathbf d_c^T\mathbf d_{c,k}
		\right).
	\end{equation}
	The far-field off-boresight angle toward user
	$k\in\mathcal K_c$ is
	$\vartheta_{j,c,k}^{\rm FF}
	=\arccos(\mathbf d_j^T\mathbf d_{c,k})$.
	By the spherical triangle inequality, we get
	\begin{equation}
		\vartheta_{j,c,k}^{\rm FF}
		\geq
		\underline{\vartheta}_{j,c}^{\rm FF}
		\triangleq
		[\Theta_{j,c}-\Delta_c]_+.
	\end{equation}
	Therefore, the worst-case RA amplitude leakage from subarray
	$j$ toward group $c$ is upper-bounded by
	\begin{equation}
		\eta_{j, c}^{\rm RA,FF}
		\triangleq
		\max_{k\in\mathcal K_c}
		\left[
		\cos
		\left(
		\vartheta_{j,c,k}^{\rm FF}
		\right)
		\right]_+^p
		\leq
		\left[
		\cos
		\left(
		\underline{\vartheta}_{j,c}^{\rm FF}
		\right)
		\right]_+^p.
	\end{equation}
	
	\textbf{2) Near-field regime:}
	In the near-field regime, the directions toward the group
	centers and individual users vary across the antennas. For
	$m\in\mathcal S_j$, we define
	\begin{equation}
		\mathbf d_{m,c}
		=
		\frac{\mathbf u_c-\mathbf t_m}
		{\|\mathbf u_c-\mathbf t_m\|_2},
		\mathbf d_{m,c,k}
		=
		\frac{\mathbf u_{c,k}-\mathbf t_m}
		{\|\mathbf u_{c,k}-\mathbf t_m\|_2}.
	\end{equation}
	Since the $m$-th RA in subarray $j$ is steered toward the
	center of its assigned group, we have
	$\mathbf f_{m,j}^{\star}=\mathbf d_{m,j}$. We define
	\begin{equation}
		\Theta_{m,j,c}
		\!\!=\!
		\arccos \!
		\left(\!
		\mathbf d_{m,j}^{T}\mathbf d_{m,c}
		\!\right),
		\Delta_{m,c}
		\!\!=\!
		\max_{k\in\mathcal K_c}
		\arccos \!
		\left(
		\mathbf d_{m,c}^{T}\mathbf d_{m,c,k}
		\right), \nonumber
	\end{equation}
	where $\Theta_{m,j,c}$ denotes the local angular separation
	between the directions from antenna $m$ toward the centers of
	groups $j$ and $c$, while $\Delta_{m,c}$ denotes the local
	angular radius of group $c$.
	
	For user $k\in\mathcal K_c$, the near-field off-boresight
	angle at antenna $m$ is
	\begin{equation}
		\vartheta_{m,j,c,k}^{\rm NF}
		=
		\arccos
		\left(
		(\mathbf f_{m,j}^{\star})^{T}
		\mathbf d_{m,c,k}
		\right).
	\end{equation}
	Using the spherical triangle inequality, we obtain
	\begin{equation}
		\vartheta_{m,j,c,k}^{\rm NF}
		\geq
		\left[
		\Theta_{m,j,c}-\Delta_{m,c}
		\right]_+.
	\end{equation}
	Accordingly, we define the conservative group-level
	off-boresight-angle lower bound as
	\begin{equation}
		\underline{\vartheta}_{j,c}^{\rm NF}
		\triangleq
		\min_{m\in\mathcal S_j}
		\left[
		\Theta_{m,j,c}-\Delta_{m,c}
		\right]_+.
	\end{equation}
	Therefore, the worst-case RA amplitude leakage from subarray
	$j$ toward group $c$ satisfies
	\begin{align}
	\!\!	\eta_{j, c}^{\rm RA,NF}
		&\! \triangleq \!
		\max_{\substack{
				m\in\mathcal S_j,
				k\in\mathcal K_c
		}}\!\!
		\left[
		\cos
		\left(
		\vartheta_{m,j,c,k}^{\rm NF}
		\right)
		\right]_+^p \!\! \leq\!\!
		\left[
		\cos
		\left(
		\underline{\vartheta}_{j,c}^{\rm NF}
		\right)
		\right]_+^p.
	\end{align}

	Based on the above far-field and near-field characterizations,
	the conservative group-level angular-separation lower bound is
	defined as
	\begin{equation}
		\underline{\vartheta}_{j,c}^{\xi}
		=
		\begin{cases}
			[\Theta_{j,c}-\Delta_c]_+,
			& \xi={\rm FF},
			\\[1mm]
			\displaystyle
			\min_{m\in\mathcal S_j}
			[\Theta_{m,j,c}-\Delta_{m,c}]_+,
			& \xi={\rm NF}.
		\end{cases}
	\end{equation}
	Therefore, the worst-case inter-group RA amplitude leakage is upper-bounded by
	\begin{equation}
		\eta_{\rm RA}^{\rm inter,\xi}
		\triangleq
		\max_{c\in\mathcal C}
		\max_{j\neq c}
		\left[
		\cos
		\left(
		\underline{\vartheta}_{j,c}^{\xi}
		\right)
		\right]_+^p,
		\xi\in\{{\rm FF},{\rm NF}\}.
	\end{equation}
	Applying $\Vert\mathbf H\Vert_2\leq\Vert\mathbf H\Vert_{\rm F}$, the off-diagonal LoS channel component satisfies
	$\left\Vert\mathbf H_{c,j}\right\Vert_2
		\leq
		\alpha_{c,j}^{\xi}
		{\eta}_{\rm RA}^{\rm inter,\xi},
		 c\neq j$,
	where $\alpha_{c,j}^{\xi}$ characterizes the path-loss- and antenna-number-dependent baseline block gain with
	\begin{equation}
		\alpha_{c,j}^{\xi}
		=
		\begin{cases}
			\displaystyle
			\left(
			\beta_0G_0N_j
			\sum\nolimits_{k\in\mathcal K_c}R_{c,k}^{-2}
			\right)^{1/2},
			& \xi={\rm FF},\\
			\displaystyle
			\left(
			\beta_0G_0
			\sum\nolimits_{k\in\mathcal K_c}
			\sum\nolimits_{m\in\mathcal S_j}
			r_{m,c,k}^{-2}
			\right)^{1/2},	& \xi={\rm NF}. \nonumber 
		\end{cases}
	\end{equation}
	Consequently, the
	inter-group leakage matrix satisfies
	\begin{equation}
		\|{\bf E}_{\rm inter}\|_2
		\le
		\alpha_{\rm inter}^{\xi}
		{\eta}_{\rm RA}^{\rm inter,\xi},
	\end{equation}
	where $\alpha_{\rm inter}^{\xi}
	\triangleq
	(
	\sum_{c\in\mathcal C}
	\sum_{j\neq c}
	(\alpha_{c,j}^{\xi})^2
	)^{1/2}$ collects the aggregate
	off-diagonal block gains. 
	Therefore, a sufficient low-leakage condition for applying Lemma \ref{pro2} is
	\begin{equation}
		\!\!\!	\alpha_{\rm inter}^{\xi}{\eta}_{\rm RA}^{\rm inter,\xi}
		\le
		\rho \sigma_{\min}(\mathbf H_{\rm bd}),
		0<\rho<1,
		\xi\in\{\mathrm{FF},\mathrm{NF}\}.
	\end{equation}
	Under this condition, the global reduced channel matrix is
	approximately block diagonal, and the global ZF precoder
	approaches the collection of group-wise ZF precoders. 
	Thus, increasing the directivity parameter $p$ reduces the normalized off-boresight RA leakage factor in both regimes. In particular, under the clipped RA pattern, $\vartheta_{j,c}^{\xi}\geq\pi/2$, $\xi\in\{{\rm FF},{\rm NF}\}$, yields zero LoS leakage from subarray $j$ to group $c$. Therefore, RA directionality suppresses the residual sidelobe and grating-lobe leakage induced by sparse aperture deployment.
	\begin{remark}[Approximate Group-Wise Beamforming Structure]
		\label{remark:groupwise_structure}
		\upshape
		Combining Lemmas~ \ref{pro1} and \ref{pro2}, when inter-group coupling and
		intra-group channel correlation are both sufficiently small, the global ZF
		beam directions can be approximated by the group-wise MRT directions.
		This observation highlights the capability of the SRA architecture in
		jointly exploiting RA directionality for inter-group leakage suppression
		and sparse-aperture selection for intra-group channel orthogonality.
	\end{remark}

\vspace{-2mm}
	To illustrate the RA-induced grating-lobe suppression, we consider a far-field uniformly sparse ULA, where the element-wise pointing directions approximately reduce to a common group-center direction $\psi_c^\star$. Its effective power pattern is
\[
P_c^{\rm eff}(\psi)
=
P_c^{\rm geo}(\psi)
\left[\cos(\psi-\psi_c^\star)\right]_{+}^{2p},
\]
where $P_c^{\rm geo}(\psi)$ denotes the geometry-only sparse-array pattern. For an inter-element spacing $d_s=q\lambda/2$ and broadside pointing $\psi_c^\star=0^\circ$, the first-order grating-lobe offset and its residual RA power factor are
\[
\sin \psi_{c,\ell}^{\rm GL}
=
\sin \psi_c^\star+\ell\frac{\lambda}{d_s}, \Delta\psi_{\rm GL}=\arcsin(2/q),
\ell= \pm 1.
\]
	Since the normalized geometry-only power equals one at a grating lobe, the corresponding residual RA amplitude and power factors are
	\begin{align}
		\eta_{\rm RA}^{\rm GL,amp}
		&=
		\left(1-{4}/{q^2}\right)^{p/2},
		\
		\eta_{\rm RA}^{\rm GL,pow}
		&=
		\left(1-{4}/{q^2}\right)^p.
	\end{align}
	Let $\eta_{\rm GL}^{\rm th}\in(0,1)$ denote the maximum tolerable normalized grating-lobe power relative to the main-lobe peak.
	To ensure
	$\eta_{\rm RA}^{\rm GL,pow}\leq\eta_{\rm GL}^{\rm th}$,
	the RA directivity and sparse-spacing factor should satisfy
	\begin{equation}
		p
		\geq
		\frac{\ln\eta_{\rm GL}^{\rm th}}
		{\ln\left(1-{4}/{q^2}\right)}, \ \text{or} \
		q
		\leq
		\frac{2}
		{\sqrt{1-\left(\eta_{\rm GL}^{\rm th}\right)^{1/p}}}.
	\end{equation}
	Fig.~\ref{fig:GL_RA_attenuation} illustrates the first-order
	grating-lobe offset and normalized RA gain versus the spacing
	factor $q$ for $p=3$. As $q$ increases, the grating lobe
	approaches the boresight and its normalized RA gain increases
	toward $0$~dB. This result reveals a tradeoff
		between aperture enlargement and grating-lobe suppression.
		For a fixed number of antennas, a larger sparse spacing
		enlarges the effective aperture but shifts the grating lobe
		closer to the RA main beam, thereby increasing its residual
		gain.

	\begin{figure}[t]
		\centering
		\includegraphics[width=0.4\textwidth]{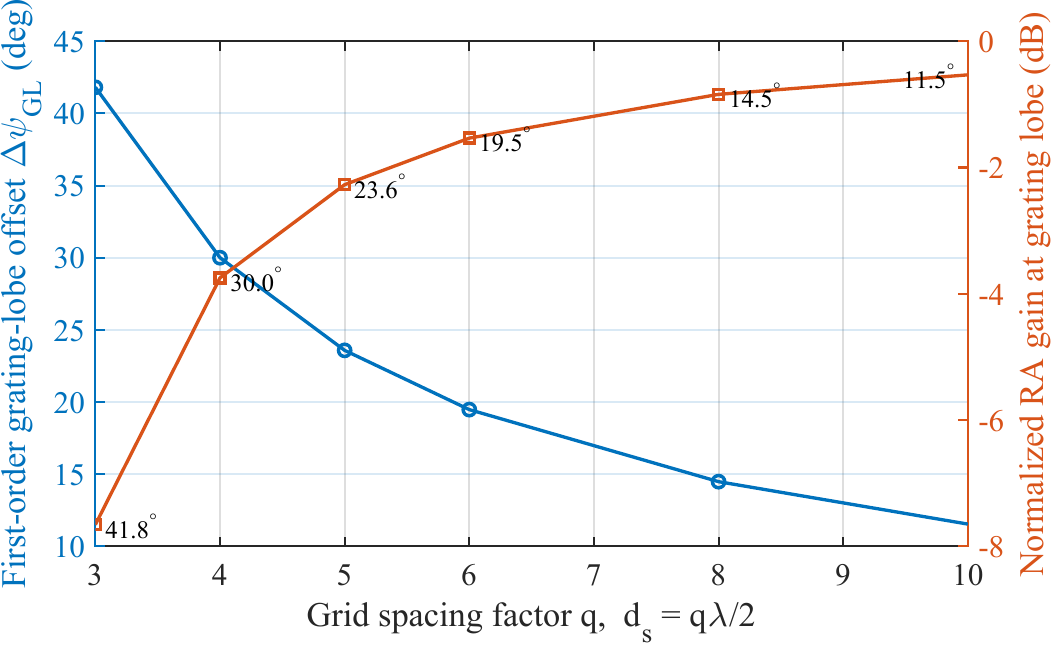}
	\caption{First-order grating-lobe offset and normalized RA gain versus the spacing factor $q$, where the RA gain is evaluated at the grating-lobe direction, $d_s=q\lambda/2$, and $p=3$.}
		\label{fig:GL_RA_attenuation}
			\vspace{-6mm}
	\end{figure}
	
	\vspace{-4mm}
\subsection{Principle-Guided Sparse Aperture Configuration}
\vspace{-1mm}

In this subsection, we develop SRA configuration guidelines for group-wise antenna-number allocation and sparse-position selection based on the preceding spatial analyses.

	\subsubsection{Approximate Antenna-Number Allocation}
	When the intra-group channels are nearly orthogonal and the inter-group leakage is weak, under approximately equal intra-group power allocation and a common noise variance $\sigma^2$, the worst-user SINR of group $c$ can be approximated as
	\begin{equation}
		\label{gammac}
		\frac{\gamma_c^{\min}}{\Gamma_c}
		\approx
		\frac{P_cN_c}{a_c\sigma^2},
		a_c\triangleq\frac{\Gamma_cK_c}{\bar g_c},
	\end{equation}
	where $P_c$ and $\bar g_c$ denote the group-level transmit power and worst-user reference per-antenna gain, respectively, with 
	\begin{equation}
		\bar g_c
		\triangleq
		\min_{k\in\mathcal K_c}
		\frac{1}{M}
		\sum_{m=1}^{M}
		\frac{\beta_0G_0}{r_{m,c,k}^{2}}
		\left[
		(\mathbf f_{m,c}^{\star})^{T}
		\mathbf d_{m,c,k}
		\right]_{+}^{2p}.
	\end{equation} 
	Under approximation \eqref{gammac}, equalizing the weighted SINR levels and using $\sum_c P_c=P_{\max}$ gives
$t
		=
		\frac{P_{\max}}
		{\sigma^2\sum_{c\in\mathcal C}a_c/N_c}$.
	Thus, the continuous antenna-number allocation is obtained by minimizing $\sum_c a_c/N_c$ subject to $\sum_cN_c=N$. The corresponding Karush-Kuhn-Tucker (KKT) solution is
	\begin{equation}
		N_c^{\mathrm{cont}}
		=
		\frac{\sqrt{a_c}}
		{\displaystyle\sum\nolimits_{j\in\mathcal C}\sqrt{a_j}}N,
		a_c=\frac{\Gamma_cK_c}{\bar g_c}.
		\label{Ncont}
	\end{equation}
	Thus, more antennas are assigned to groups with higher priority, larger user loads, or weaker reference gains.

	The allocation must also satisfy the minimum array-size and spatial-DoF requirements. Defining
$N_c^{\rm req}
		=
		\max\{N_c^{\min},K_c\}$,
	the constrained continuous allocation is written as
	\begin{align}
		\label{Nc}
		N_c^\star
		=
		\max\left\{
		N_c^{\rm req},
		\sqrt{{a_c}/{\mu}}
		\right\},
	\end{align}
	where $\mu$ is chosen such that
	$\sum_{c=1}^{C}N_c^\star=N$. The spatial-resolution requirement is incorporated through the target aperture length $D_{c,x}^{\star}$ in the subsequent position initialization. Finally, $\{N_c^\star\}$ are rounded to integers while preserving the total antenna number.
	
	\subsubsection{Approximate Antenna-Position Allocation}
	
	After determining $\{N_c\}_{c=1}^{C}$, the position sets $\{\mathcal S_c\}_{c=1}^{C}$ are initialized to provide sufficient intra-group resolution while mitigating sparse-aperture ambiguity and inter-group leakage. We present the construction for a ULA along the $x$-axis, which can be extended to a uniform planar array (UPA) along the dominant resolution dimension or over a two-dimensional aperture.
	
	According to the far-field and near-field orthogonalization conditions in Section~III-A, the required aperture of group $c$ is selected as
	\begin{equation}
		D_{c,x}^{\rm req}
		=
		\begin{cases}
			D_{c,x}^{\rm th,FF}, & \text{far-field design},\\
			D_{c,x}^{\rm th,NF}, & \text{near-field design}.
		\end{cases}
	\end{equation}
	Let
	$D_{\max}=(M-1)\Delta$
	denote the maximum candidate-aperture length, and let
	$D_{c,\min}=(N_c-1)\Delta$
	denote the minimum span required for $N_c$ distinct positions. The target aperture is then
	\begin{align}
		D_{c,x}^{\star}
		=
		\min\left\{
		D_{\max},
		\max\left\{
		D_{c,\min},
		D_{c,x}^{\rm req}
		\right\}
		\right\}.
	\end{align}
	Given $N_c$ and $D_{c,x}^{\star}$, the continuous positions are generated by
	\begin{equation}
		\label{AntennaPosition}
	\!	x_{c,n}^{\rm cont}
		\!=\!
		\!-\! \frac{D_{c,x}^{\star}}{2}
		\!+\!
		\frac{(n-1)D_{c,x}^{\star}}{N_c \!\!-\!\!1}
		\!+\!
		\eta_{\rm jit}
		\frac{D_{c,x}^{\star}}{N_c \!\!-\!\! 1}
		\omega_{c,n},
		 n \!= \!1,\ldots,N_c,
	\end{equation}
	where
	\begin{equation}
		\omega_{c,n}
		=
		\begin{cases}
			0, & n=1 \ \text{or} \ n=N_c,\\
			2\{(n-1)\phi_g+s_c\}-1,
			& 2\leq n\leq N_c-1,
		\end{cases}
		\label{jitter_sequence}
	\end{equation}
where $\{\cdot\}$ denotes the fractional part, $\phi_g=(\sqrt{5}-1)/2$ is the golden-ratio constant, $s_c$ is a group-dependent shift, and $\eta_{\rm jit}\in(0,1/2)$ controls the perturbation strength. This construction spreads the antennas over the target aperture while avoiding strictly periodic sparse layouts.
	Let $\mathcal M=\{1,\ldots,M\}$ denote the candidate-position index set. Each continuous position is projected onto the nearest unoccupied grid point according to
	\begin{align}
		\label{Projection}
		m_{c,n}
		=
		\arg\min_{m\in\mathcal M_{\rm free}}
		\left|
		\left(m-\frac{M+1}{2}\right)\Delta
		-
		x_{c,n}^{\rm cont}
		\right|^2,
	\end{align}
	where $\mathcal M_{\rm free}\subseteq\mathcal M$ is updated after each projection. The resulting allocation provides a physically meaningful initialization, which is subsequently evaluated and refined using the true weighted max-min SINR objective.

\vspace{-3mm}
	\section{Proposed Solution}
	\label{Proposed Solution}
	\vspace{-2mm}

In this section, we develop a practical low-complexity AO framework for problem~\eqref{P0} guided by the spatial design principles of SRA established in Section~\ref{Performance Analysis}. To reduce the coupled optimization complexity among sparse aperture allocation, RA orientations, and transmit beamforming, we exploit the RA directionality insight in Section~\ref{III-B} and adopt a projected group-center pointing rule for the activated RAs, whose detailed design is presented in Section~\ref{IV-A}. After configuring the RA orientations based on the derived pointing rule, we introduce an auxiliary variable $\tau$ and reformulate the resulting problem as

	\begin{subequations}
			\label{P1}
			\begin{eqnarray}
					&\!\!\!\!\!\!
					\max \limits_{\mathbf B_{c},\mathbf w_{c,k}, \tau}
					&\!\!\! \tau \label{P0_obj}\\
					&\!\!\!\!\!\! \mathrm{s.t.}
					&\!\!\! \gamma_{c,k}(\mathbf B_{c},\mathbf w_{c,k})
					\geq \tau \Gamma_c, \forall c,\ k\in\mathcal K_c, \label{P0_sinr}\\
					&&\!\!\! \mathrm{(\ref{P1_power})}-\mathrm{(\ref{P1_min_ant})}.
				\end{eqnarray}
		\end{subequations} 
Problem~\eqref{P1} remains challenging because the binary allocation variables jointly determine the antenna numbers, sparse positions, projected boresights, and effective channels, leading to a combinatorial search space. Although the beamforming subproblem for a fixed allocation can be efficiently solved via bisection and SOCP \cite{BoydConvex}, exhaustive search over all allocations is computationally prohibitive. Therefore, we develop an AO framework that updates the sparse antenna allocation through a sampled multi-start search and evaluates each candidate allocation by solving the corresponding reduced-dimensional beamforming problem.

\vspace{-4mm}
	\subsection{Projected Group-Center RA Orientation Rule}
	\label{IV-A}
	\vspace{-1mm}
	
We first derive the structured RA orientation rule adopted in problem~\eqref{P1}. For a fixed antenna allocation, jointly optimizing RA orientations and transmit beamforming leads to a highly coupled design problem. To address this issue, we propose a low-complexity projected group-center orientation rule, where each activated RA maximizes its directional gain toward the center of its assigned service group subject to the practical rotation constraint. This rule is optimal for a singleton group and provides a representative pointing direction for compact multiuser groups. Moreover, as characterized in Section~\ref{III-B}, aligning the boresight directions of activated RAs with their assigned group centers reduces the effective channel gains toward non-target groups, thereby suppressing inter-group interference. 
For a given antenna allocation, the antenna set assigned to
service group $c$ is denoted by
	\begin{equation}
		\mathcal S_c=\{m:b_{m,c}=1\}, |\mathcal S_c|=N_c .
	\end{equation}
	For the $m$-th RA assigned to service group $c$, its boresight direction is set as
	\begin{equation}
		\mathbf f_{m,c}^{0}
		=
		\frac{\mathbf u_c-\mathbf t_m}
		{\Vert \mathbf u_c-\mathbf t_m\Vert_2},
		 m\in\mathcal S_c, c\in\mathcal C .
	\end{equation}
	The corresponding azimuth and zenith angles are given by
	\begin{align}
		\label{theta_closed}
			&\theta_{m,c}^{0}
	=
		\operatorname{atan2}(y_c-y_m,x_c-x_m),\\
		\label{phi_closed}
		&\phi_{m,c}^{0}
		=
		\operatorname{atan2}
	(
		\sqrt{(x_c-x_m)^2+(y_c-y_m)^2},
		z_c-z_m
	).
	\end{align}
	If the above pointing direction lies within the feasible rotation range, i.e., $\phi_{m,c}^{0}\leq \phi_{\max}$, the RA is directly steered toward the group center. Otherwise, the RA orientation is projected onto the boundary of the feasible rotation cone while keeping the same azimuth direction. Therefore, the adopted zenith angle is given by
	$\phi_{m,c}^{\star}=
		\min\{\phi_{m,c}^{0},\phi_{\max}\}$.
	The resulting feasible boresight direction is
	\begin{equation}
				\label{fm}
	\!\!\mathbf f_m \!=\!
		\mathbf f_{m,c}^{\star} \!=\!
		\left[
		\sin\phi_{m,c}^{\star}\cos\theta_{m,c}^{0},
		\sin\phi_{m,c}^{\star}\sin\theta_{m,c}^{0},
		\cos\phi_{m,c}^{\star}
		\right]^T. \nonumber 
	\end{equation}
	This projected group-center pointing rule always satisfies the practical rotation constraint \eqref{P1_fm}, and reduces to the direct group-center pointing rule when the group center is located inside the feasible rotation cone.

	In the considered system, the local departure direction from antenna $m$ to user $k\in\mathcal K_c$ is
$\mathbf d_{m,c,k}
=
\frac{\mathbf u_{c,k}-\mathbf t_m}
{\Vert \mathbf u_{c,k}-\mathbf t_m \Vert_2}$.
	To characterize the directional service quality provided by the projected group-center configuration, we define the worst-case off-boresight angle within service group $c$ as
	\begin{equation}
	\label{violation}
	\!\delta_c^{\max}
	\!=\!
	\!\!\max_{m\in\mathcal S_c, k\in\mathcal K_c}
	\arccos
	\left(
	(\mathbf f_{m,c}^{\star})^T\mathbf d_{m,c,k}
	\right),
	c \!\in \!\mathcal C.
\end{equation}
	Accordingly, the minimum relative RA power gain experienced by users in group $c$ is $\left[\cos\left(\delta_c^{\max}\right)\right]_{+}^{2p}$. For illustration, we consider a circular group of radius $6$ m whose center is $40$ m from the array reference point, with $p=3$. The reference-point approximation gives $\delta_c^{\max}\approx\arcsin(6/40)=8.63^\circ$, corresponding to a worst-case directional-gain loss of only about $0.30$ dB. 
	This indicates that the group-center pointing rule provides effective directional coverage for compact groups, validating its use as a low-complexity orientation design strategy.

	\vspace{-4mm}
	\subsection{Reduced-Dimensional Beamforming}
	\vspace{-1mm}
	For a given antenna allocation $\mathbf B$, the corresponding
	RA orientations are determined by the projected group-center
	rule in Section~\ref{IV-A}. Then, we optimize the
	transmit beamformers over the selected group-specific
	subarrays.
	By introducing the selection matrix
	$\mathbf E_c
		=
		[\mathbf e_{s_{c,1}},\ldots,\mathbf e_{s_{c,N_c}}]
		\in\{0,1\}^{M\times N_c}$,
	where $\mathcal S_c=\{s_{c,1},\ldots,s_{c,N_c}\}$ and $\mathbf e_m$ is
	the $m$-th canonical basis vector, the full-dimensional beamformer can
	 be given by
	\begin{align}
		\mathbf w_{c,k}
		=
		\mathbf E_c\mathbf v_{c,k},
		\label{eq:w_Ev}
	\end{align}
	where $\mathbf v_{c,k}\in\mathbb C^{N_c\times 1}$
	is the reduced-dimensional beamforming vector. Since
	$\mathbf E_c^H\mathbf E_c=\mathbf I_{N_c}$, we have $\|\mathbf w_{c,k}\|_2^2
	=
	\|\mathbf v_{c,k}\|_2^2$.
	Therefore, the support constraint
	$\mathbf B_c\mathbf w_{c,k}=\mathbf w_{c,k}$ is automatically satisfied
	by \eqref{eq:w_Ev}.
	For user $k$ in group $c$, we define the reduced channel from the
	subarray of group $j$ to the user as
	\begin{align}
		\tilde{\mathbf h}_{j, c,k}
		=
		\mathbf E_j^H
		\mathbf h_{c,k}(\mathbf B,\{\mathbf f_{m,\ell}^\star\})
		\in\mathbb C^{N_j\times 1}.
	\end{align}
	Then, we have $\mathbf h_{c,k}^H\mathbf w_{j,i}
	=
	\tilde{\mathbf h}_{j, c,k}^H
	\mathbf v_{j,i}$.
	The SINR of user $k$ in group $c$ can be rewritten as
	\begin{align}
		\!\!\!\! \gamma_{c,k}
		\!=\!
		\frac{
			\left|
			\tilde{\mathbf h}_{c, c,k}^H
			\mathbf v_{c,k}
			\right|^2
		}{
			\sum\limits_{i\in\mathcal K_c,i\neq k}
			\left|
			\tilde{\mathbf h}_{c, c,k}^H
			\mathbf v_{c,i}
			\right|^2
			\!\!+\!\!
			\sum\limits_{j\neq c}
			\sum\limits_{i\in\mathcal K_j}
			\left|
			\tilde{\mathbf h}_{j, c,k}^H
			\mathbf v_{j,i}
			\right|^2
			\!\!+\!
			\sigma_k^2
		}.
		\label{eq:SINR_reduced}
	\end{align}
	Fixed $\mathbf B$ and $\mathbf f_{m,c}^\star$, the reduced beamforming
	subproblem is
	\begin{subequations}
		\label{P1_W}
		\begin{eqnarray}
			&\max \limits_{\mathbf v_{c,k},\tau}
			& \tau \\
			&\mathrm{s.t.}
			& \gamma_{c,k}\geq \tau\Gamma_c,
			\forall c,\ k\in\mathcal K_c, \label{PV_sinr}\\
			&& \sum\nolimits_{c=1}^{C}\sum\nolimits_{k=1}^{\mathcal K_c}
			\|\mathbf v_{c,k}\|_2^2
			\leq P_{\max}. \label{PV_power}
		\end{eqnarray}
	\end{subequations}
	Problem \eqref{P1_W} is quasi-convex with respect to $\tau$.
	For a fixed $\tau$, the SINR constraint in \eqref{PV_sinr} can be written
	as
	\begin{align}
		|
		\tilde{\mathbf h}_{c, c,k}^H
		\mathbf v_{c,k}
		|^2
		\geq
		\tau\Gamma_c
		\left(
		\sum_{(j,i)\neq(c,k)}
		|
		\tilde{\mathbf h}_{j, c,k}^H
		\mathbf v_{j,i}
		|^2
		+
		\sigma_k^2
		\right).
	\end{align}
	Since the phase of $\mathbf v_{c,k}$ can be freely rotated without
	changing the SINR or the transmit power, we impose
	\begin{align}
		\label{SOCP-CON}
		\operatorname{Im}
		\left\{
		\tilde{\mathbf h}_{c, c,k}^H
		\mathbf v_{c,k}
		\right\}=0,
		\operatorname{Re}
		\left\{
		\tilde{\mathbf h}_{c, c,k}^H
		\mathbf v_{c,k}
		\right\}\geq 0.
	\end{align}
	Then, for fixed $\tau$, each SINR constraint is equivalent to the
	second-order cone (SOC) constraint
	\begin{align}
		\left\|
		\begin{bmatrix}
			\left\{
			\tilde{\mathbf h}_{j, c,k}^H
			\mathbf v_{j,i}
			\right\}_{(j,i)\neq(c,k)}
			\\
			\sigma_k
		\end{bmatrix}
		\right\|_2
		\!\! \leq \!\!
		\frac{1}{\sqrt{\tau\Gamma_c}}
		\! \operatorname{Re}
		\left\{ \!
		\tilde{\mathbf h}_{c, c,k}^H
		\mathbf v_{c,k}
		\! \right\}.
		\label{eq:SOCP_constraint}
	\end{align}
	Therefore, for a given $\tau$, the feasibility problem is
	\begin{subequations}
		\label{P1_W1}
		\begin{eqnarray}
			&\!\!\!\!\!\!\!\!\!\!\!\! \mathrm{find}
			& \{\mathbf v_{c,k}\}\\
			&\!\!\!\!\!\!\!\!\!\!\!\! \mathrm{s.t.}
			& \eqref{PV_power},\eqref{SOCP-CON},\eqref{eq:SOCP_constraint},
			\forall c, k\in\mathcal K_c.
		\end{eqnarray}
	\end{subequations}
	Problem \eqref{P1_W1} is a standard SOCP feasibility problem. Hence, the optimal $\mathbf{w}_{c,k}$ for a given $\mathbf B$ can be obtained by
	bisection over $\tau$. The initial bisection interval can be set as
	\begin{align}
		\tau_{\min}=0,
		\tau_{\max}
		=
		\min_{c,k}
		{
			P_{\max}\|\tilde{\mathbf h}_{c, c,k}\|_2^2
		}/{
			\Gamma_c\sigma_k^2
		},
	\end{align}
	where $\tau_{\max}$ is an interference-free upper bound. The lower bound $\tau_{\min}=0$ is used only for initializing the bisection,
	and the SOCP feasibility test is performed for $\tau>0$.
	
	\vspace{-4mm}
\subsection{Sampled Multi-Start Antenna Allocation}
\vspace{-1mm}
	Let $\tau(\mathbf B)$ denote the optimal weighted max-min
	SINR returned by beamforming problem \eqref{P1_W} for a
	feasible allocation $\mathbf B$. Then, we optimize the
	antenna allocation as
	\begin{eqnarray}
		\label{P1_B}
		&
		\max \limits_{\mathbf B}
		\ \tau(\mathbf B) 
		\quad \mathrm{s.t.}
		\ \eqref{P1_total_ant}-\eqref{P1_min_ant}.
	\end{eqnarray}
To address problem \eqref{P1_B}, we adopt a sampled multi-start greedy search \cite{Yun2026}.
	For the $\ell$-th start, the initial allocation $\mathbf B_{\ell}^{(0)}$ is generated using the antenna-number rules in \eqref{Ncont}-\eqref{Nc} and the non-periodic position initialization in \eqref{AntennaPosition}-\eqref{Projection}. The same analysis-guided antenna numbers are retained across different starts, while distinct initial position patterns are generated by independently selecting the group-dependent shifts $s_c$ in \eqref{jitter_sequence}. For notational simplicity, the start index $\ell$ is omitted during each local search, and $\mathbf B^{(r)}$ denotes the current allocation at outer iteration $r$.
	Let ${\cal S}_c^{(r)}
		=
		\{
			m:b_{m,c}^{(r)}=1
			\},
		N_c^{(r)}
		=
		|{\cal S}_c^{(r)}|$,
	and define the currently unoccupied candidate-position set as
	\begin{equation}
		{\cal M}_{\rm free}^{(r)}
		=
		\{1,\ldots,M\}
		\setminus
		\bigcup\nolimits_{c=1}^{C}{\cal S}_c^{(r)}.
	\end{equation}
	At iteration $r$, a candidate allocation $\widetilde{\mathbf B}$ is generated from $\mathbf B^{(r)}$ through one of the following three local moves. The antenna sets associated with $\widetilde{\mathbf B}$ are denoted by $\{\widetilde{\cal S}_c\}_{c=1}^{C}$. The corresponding binary allocation is defined by $\widetilde b_{m,c}=1$ if $m\in\widetilde{\mathcal S}_c$ and $\widetilde b_{m,c}=0$ otherwise.
	
	The replacement move substitutes one occupied position of group $c$ with an unoccupied position while preserving $N_c^{(r)}$
	\begin{equation}
		\widetilde{\cal S}_c
		=
		{\cal S}_c^{(r)}
		\setminus
		\{m_{\rm old}\}
		\cup
		\{m_{\rm new}\},
		\widetilde{\cal S}_i
		=
		{\cal S}_i^{(r)}, i\ne c,
	\end{equation}
	where $m_{\rm old}\in{\cal S}_c^{(r)}$ and $m_{\rm new}\in{\cal M}_{\rm free}^{(r)}$.
	The swap move exchanges one occupied position between groups $c$ and $j$
	\begin{align}
		\widetilde{\cal S}_c
		&=
		{\cal S}_c^{(r)}
		\setminus
		\{m_c\}
		\cup
		\{m_j\},
		\widetilde{\cal S}_j
		=
		{\cal S}_j^{(r)}
		\setminus
		\{m_j\}
		\cup
		\{m_c\},
	\end{align}
	where $m_c\in{\cal S}_c^{(r)}$, $m_j\in{\cal S}_j^{(r)}$, and $c\ne j$.
	The reallocation move transfers one antenna from group $c$ to group $j$ when $N_c^{(r)}>N_c^{\rm req}$
	\begin{align}
		\widetilde{\cal S}_c
		&=
		{\cal S}_c^{(r)}
		\setminus
	\{m_c\},
		\widetilde{\cal S}_j
		=
		{\cal S}_j^{(r)}
		\cup
		\{m_c\},
	\end{align}
	where $m_c\in{\cal S}_c^{(r)}$ and $c\ne j$. This move changes the group-wise antenna numbers from $(N_c^{(r)},N_j^{(r)})$ to $(N_c^{(r)}-1,N_j^{(r)}+1)$. For the swap and reallocation moves, all unspecified group sets remain unchanged.
Candidates that violate constraints~\eqref{P1_total_ant}--\eqref{P1_min_ant} or duplicate another allocation in the current sampled neighborhood are discarded.
	Let ${\cal N}_{\rm rep}^{(r)}$, ${\cal N}_{\rm swap}^{(r)}$, and ${\cal N}_{\rm rea}^{(r)}$ denote the retained candidate allocations generated by the replacement, swap, and reallocation moves, respectively. The sampled neighborhood of $\mathbf B^{(r)}$ is then
	\begin{align}
		&{\cal N}_L\bigl(\mathbf B^{(r)}\bigr)
		=
		{\cal N}_{\rm rep}^{(r)}
		\cup
		{\cal N}_{\rm swap}^{(r)}
		\cup
		{\cal N}_{\rm rea}^{(r)},\\
		&\left|{\cal N}_{\rm rep}^{(r)}\right|
		\leq L_{\rm rep},
		\left|{\cal N}_{\rm swap}^{(r)}\right|
		\leq L_{\rm swap},
		\left|{\cal N}_{\rm rea}^{(r)}\right|
		\leq L_{\rm rea},\\
		&\left|
		{\cal N}_L\bigl(\mathbf B^{(r)}\bigr)
		\right|
		\leq
		L_{\rm neigh}
		=
		L_{\rm rep}
		+
		L_{\rm swap}
		+
		L_{\rm rea}.
	\end{align}
	Each candidate $\widetilde{\mathbf B}\in{\cal N}_L(\mathbf B^{(r)})$ is evaluated using the bisection-SOCP solver. The best candidate at iteration $r$ is selected as $
		\widehat{\mathbf B}^{(r)}
		=
		\arg\max_{\widetilde{\mathbf B}
			\in
			{\cal N}_L(\mathbf B^{(r)})}
		\tau(\widetilde{\mathbf B})$.
	To avoid accepting numerical fluctuations, $\epsilon_{\rm acc}>0$ is introduced. The allocation is updated only if the objective improvement exceeds $\epsilon_{\rm acc}$
	\begin{equation}
		\mathbf B^{(r+1)}
		=
		\begin{cases}
			\widehat{\mathbf B}^{(r)},
			&
			\tau\bigl(\widehat{\mathbf B}^{(r)}\bigr)
			>
			\tau\bigl(\mathbf B^{(r)}\bigr)
			+
			\epsilon_{\rm acc},\\
			\mathbf B^{(r)},
			&
			\text{otherwise}.
		\end{cases}
	\end{equation}
	The current local search terminates when the sampled neighborhood is empty, no improving candidate is found, or $I_{\max}^{\rm out}$ iterations are reached.
	
	The above procedure is repeated for $S_{\max}$ starts. Let $\mathbf B_{\ell}^{\rm loc}$ denote the final allocation obtained from the $\ell$-th start. The best start is selected as
	$\ell^\star
		=
		\arg\max_{\ell=1,\ldots,S_{\max}}
		\tau\bigl(\mathbf B_{\ell}^{\rm loc}\bigr)$
	and the final allocation is $\mathbf B^\star
		=
		\mathbf B_{\ell^\star}^{\rm loc}$.
	The beamforming problem \eqref{P1_W1} is then re-solved under $\mathbf B^\star$ to obtain ${\mathbf v_{c,k}^\star}$, and the full-dimensional beamformers are recovered as
	\begin{equation}
		\mathbf w_{c,k}^\star
		=
		\mathbf E_c(\mathbf B^\star)
		\mathbf v_{c,k}^\star,
		\forall c\in{\cal C},\ k\in{\cal K}_c.
	\end{equation}
	By optimizing the beamformers in the reduced dimensions $\{N_c\}$, the proposed framework substantially lowers the cost of the SOCP-based optimization.

\vspace{-4mm}
	\subsection{Convergence and Computational Complexity Analysis}
	\vspace{-1mm}
	For each accepted antenna-allocation update, the acceptance criterion guarantees
	$\tau\big(\mathbf B^{(r+1)}\big)
		>
		\tau\big(\mathbf B^{(r)}\big)
		+\epsilon_{\rm acc}$.
	Therefore, the objective value increases monotonically along each local search. Since the feasible antenna-allocation set is finite, each search terminates after a finite number of accepted updates. The multi-start strategy preserves this finite-termination property while reducing the risk of obtaining a poor local solution.
	
	We next analyze the computational complexity. For a fixed allocation $\mathbf B$, let
	$K=\sum\nolimits_{c=1}^{C}K_c$ and $
		n_v(\mathbf B)
		=
		2\sum\nolimits_{c=1}^{C}K_cN_c$
	denote the number of users and real beamforming variables, respectively. For a fixed $\tau$, problem~\eqref{P1_W1} contains $K$ SINR SOC constraints of dimension $\mathcal O(K)$ and one total-power SOC constraint of dimension $\mathcal O(n_v(\mathbf B))$. Therefore, the complexity of one interior-point feasibility test is
	\begin{equation}
		C_{\rm IP}(\mathbf B)
		=
		\mathcal O\left(
		I_{\rm IP}
		\left[
		n_v^3(\mathbf B)
		+
		n_v^2(\mathbf B)K^2
		+
		n_v(\mathbf B)K^3
		\right]
		\right),
	\end{equation}
	where $I_{\rm IP}$ denotes the number of interior-point iterations. With bisection accuracy $\epsilon_{\rm bis}$, the number of feasibility tests is
	$T_{\rm bis}
		=
		\left\lceil
		\log_2
		\frac{\tau_{\max}-\tau_{\min}}
		{\epsilon_{\rm bis}}
		\right\rceil$,
	and the cost of evaluating one candidate allocation is
	$C_{\rm eval}(\mathbf B)
		=
		T_{\rm bis}C_{\rm IP}(\mathbf B)$, where $\lceil x \rceil$ denotes the smallest integer not less than $x$.
	Since at most $L_{\rm neigh}$ candidates are evaluated in each outer iteration, the overall worst-case complexity is
	\begin{equation}
		C_{\rm prop}
		=
		\mathcal O\left(
		\left[
		S_{\max}
		\left(
		1+I_{\max}^{\rm out}L_{\rm neigh}
		\right)
		+1
		\right]
		\bar C_{\rm eval}
		\right),
	\end{equation}
	where $\bar C_{\rm eval}
		=
		\max_{\mathbf B\in\mathcal B}
		C_{\rm eval}(\mathbf B)$.
	Thus, the computational cost is mainly determined by the number of sampled allocations and the dimension of the reduced beamforming problem. The sampled neighborhood limits the former, while the selection-matrix formulation reduces the latter from the full aperture dimension $M$ to ${N_c}$.

	\vspace{-2mm}
	\section{Simulation Results}
	\label{Simulation Results}
	\vspace{-1mm}

	\subsection{Simulation Setup}
	\vspace{-1mm}

	We consider an SRA-enabled multi-group communication system, where the BS is located at the origin and the users are distributed on the vertical $xOz$ plane.
 The carrier frequency is $f_c=28$ GHz, corresponding to wavelength $\lambda=10.7$ mm. The BS is equipped with a candidate ULA aperture with $M=320$ positions and inter-grid spacing $\Delta=\lambda/2$ \cite{Ahmed2026}. Among them, $N=96$ RAs are activated, with the minimum antenna number of each group set to $N_c^{\min}=K_c$. Unless otherwise specified, the maximum transmit power is $P_{\max}=30$ dBm, the noise power is $\sigma_k^2=-80$ dBm, the RA directivity factor is $p=3$, and the maximum zenith angle is $\phi_{\max}=\pi/2$.

	We adopt a general non-uniform user distribution over the coverage region $\mathcal A$, partitioned into $C$ service groups. A multiuser group represents a compact hotspot, whereas a spatially isolated user is modeled as a singleton group with $K_c=1$ and $\mathbf u_c=\mathbf u_{c,1}$. For each compact group, the center is given by $\mathbf u_c=L_c[\sin\psi_c,0,\cos\psi_c]^T$, where $L_c$ and $\psi_c$ denote its distance and signed angular direction in the $xOz$ plane, respectively, and its users are uniformly distributed within a local disk of radius $R_c$.
	Unless otherwise specified, we consider $C=3$ compact hotspot groups with $(L_1,L_2,L_3)=(40,40,40)$ m, $(\psi_1,\psi_2,\psi_3)=(-60^\circ,0^\circ,60^\circ)$, and $(K_1,K_2,K_3)=(2,3,4)$. The hotspot radius is $R_c=6$ m. Each group is associated with $D_{\rm c}=5$ local scatterers located within the propagation corridor between the BS and the group center. 
	All results are averaged over $50$ independent channel and user-location realizations. For the proposed sampled multi-start greedy search, we set $S_{\max}=10$, $I_{\max}^{\rm out}=30$, $\epsilon_{\rm acc}=10^{-3}$, and $\epsilon_{\rm bis}=10^{-4}$. The sampled neighborhood sizes are $L_{\rm rep}=50$, $L_{\rm swap}=30$, and $L_{\rm rea}=20$, yielding at most $L_{\rm neigh}=100$ candidates per iteration. 

We consider the following seven benchmark schemes:
1) the initialization-only SRA scheme directly applies the analysis-guided antenna-number and non-periodic position initialization in \eqref{Ncont}--\eqref{Projection} without greedy refinement;
2) the equal-allocation SRA scheme evenly distributes the activated RAs among the groups and optimizes their positions using the same sampled multi-start search;
3) the compact-subarray scheme assigns each group a contiguous subarray with the same group-wise antenna numbers as the proposed SRA design;
4) the uniform-position SRA scheme adopts regularly spaced group-specific sparse positions with the same group-wise antenna numbers as the proposed SRA design;
5) the omni sparse-array scheme sets $p=0$;
6) the fixed-orientation SRA scheme uses the proposed allocation with $\mathbf f_m=\mathbf e_z$ for all activated RAs; and
7) the random-allocation SRA scheme randomly generates sparse-aperture allocations and reports the best feasible solution among $500$ random realizations, where the corresponding RA orientations and transmit beamforming are optimized for each allocation.

	\vspace{-2mm}
	\subsection{Performance Analysis}
	\vspace{-1mm}
	\begin{figure}[t]
		\centering
		\includegraphics[width=0.36 \textwidth]{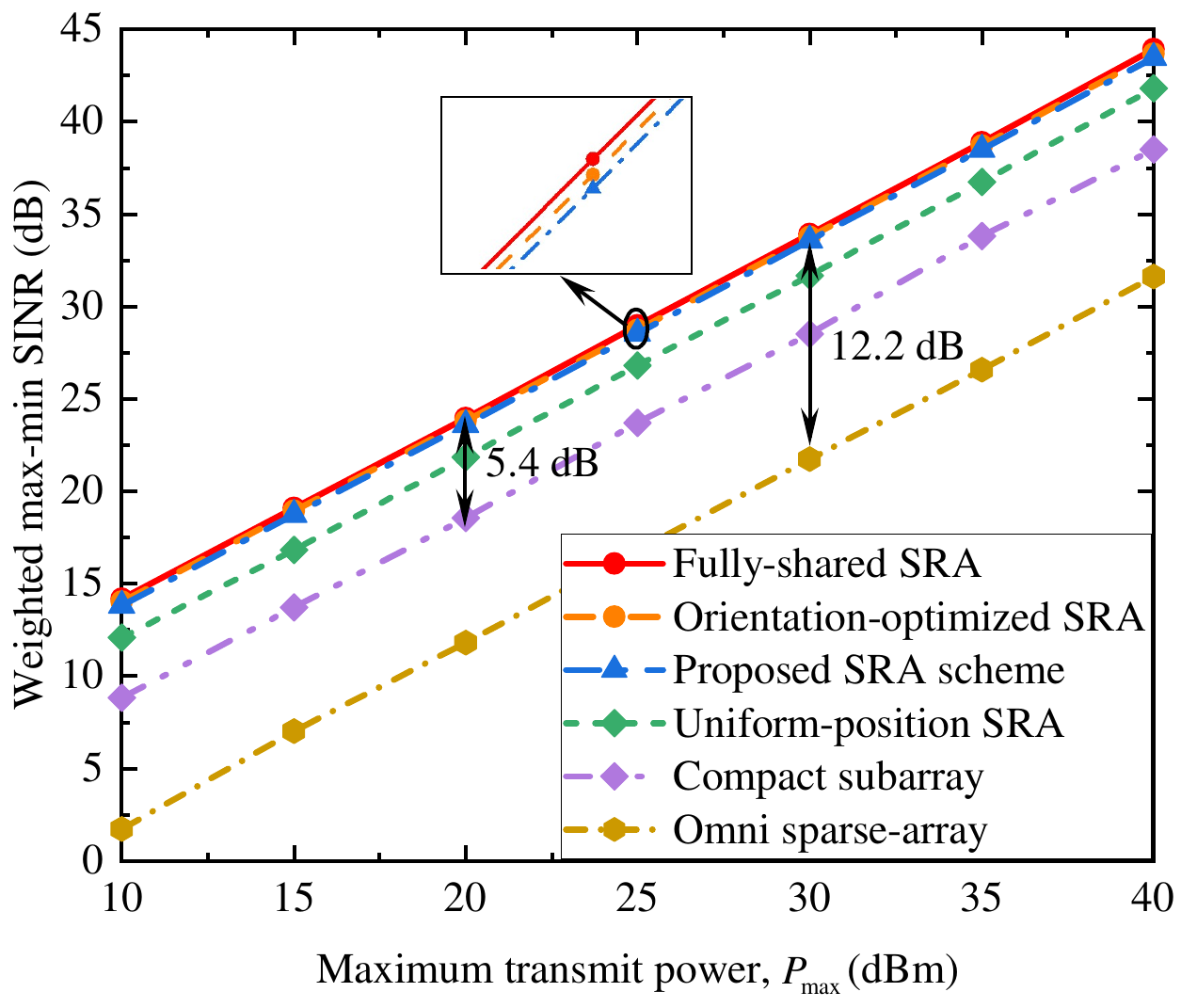}
	\caption{Weighted max-min SINR versus maximum transmit power $P_{\rm max}$.}
		\label{FigP}
		\vspace{-6mm}
	\end{figure}

Fig.~\ref{FigP} shows the weighted max-min SINR versus maximum
transmit power $P_{\max}$. For comparison, the fully-shared
SRA scheme allows activated RAs to jointly serve all users
over a shared sparse aperture with optimized orientations,
whereas the orientation-optimized SRA scheme uses the
proposed antenna allocation while optimizing the RA orientations
and transmit beamforming. The proposed
SRA scheme performs close to the fully-shared SRA scheme. Although
the fully-shared SRA scheme enables greater full-array
cooperation, RA-induced isolation weakens inter-group coupling,
enabling the proposed group-specific architecture to retain most
of its performance. Moreover, the negligible gap from the
orientation-optimized SRA scheme indicates that the element-wise
group-center pointing rule is effective for compact groups,
avoiding additional orientation optimization with little
performance loss.
The proposed SRA scheme outperforms the compact subarray scheme.
At $P_{\max}=20$ dBm, it achieves a $5.4$-dB gain in the
weighted max-min SINR, demonstrating the importance of enlarging
the effective sparse aperture for spatial resolution. It also
achieves a $12.2$-dB gain over the omni sparse-array scheme at
$P_{\max}=30$ dBm. Although the omni sparse-array scheme employs
sparse activation, it cannot exploit RA directionality to
suppress inter-group leakage.

		\begin{figure}[t]
		\centering \includegraphics[width=0.36\textwidth]{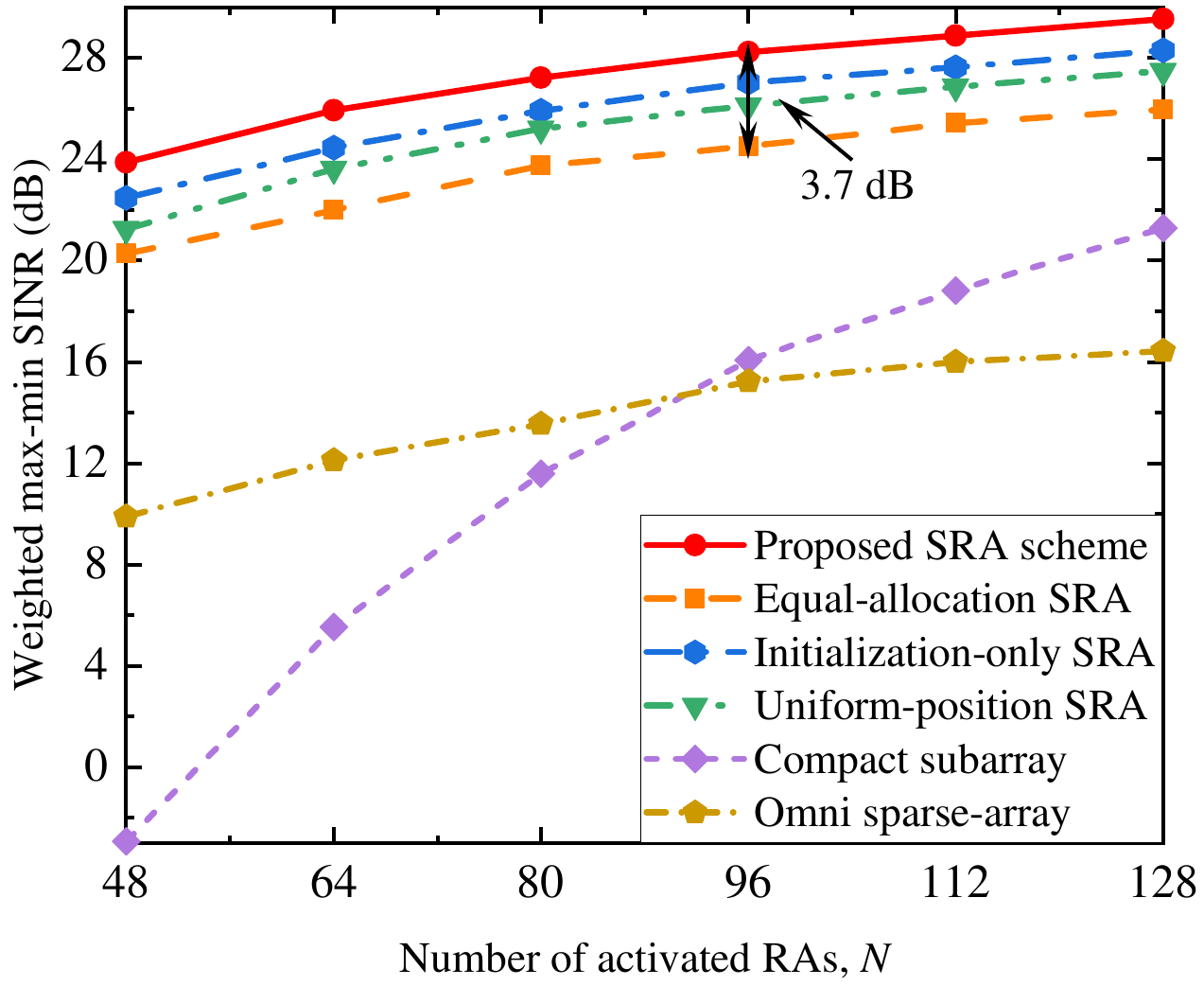}
	\caption{Weighted max-min SINR versus the number of
		activated RAs $N$.}
		\label{fig:N_case}
		\vspace{-4mm}
	\end{figure}

	Fig.~\ref{fig:N_case} compares the weighted max-min SINR
	versus the number of activated RAs $N$ in a mixed deployment
	with $(K_1,K_2,K_3)=(1,3,8)$ and
	$\boldsymbol{\Gamma}=(1,2,5)$. The weighted max-min SINRs
	of all schemes increase with $N$ owing to the enhanced array
	gain, enlarged beamforming design space, and greater
	flexibility in group-wise resource allocation. The proposed SRA
	scheme consistently achieves the highest weighted fairness
	performance by jointly adapting the group-wise antenna
	numbers and non-periodic sparse positions to the heterogeneous
	user loads, priority weights, and channel conditions. At
	$N=96$, it outperforms the equal-allocation SRA scheme by
	approximately $3.7$-dB, demonstrating the importance of
	load- and priority-aware aperture-resource allocation. The
	initialization-only SRA scheme also performs competitively,
	confirming the effectiveness of the analysis-guided
	initialization, while the subsequent sampled search provides
	an additional gain. In contrast, the uniform-position SRA scheme
	is limited by its periodic placement, the compact subarray
	scheme suffers from insufficient aperture resolution, and the
	omni sparse-array scheme cannot exploit RA-induced inter-group
	leakage suppression.

	\begin{figure}[t]
		\centering
		\includegraphics[width=0.36 \textwidth]{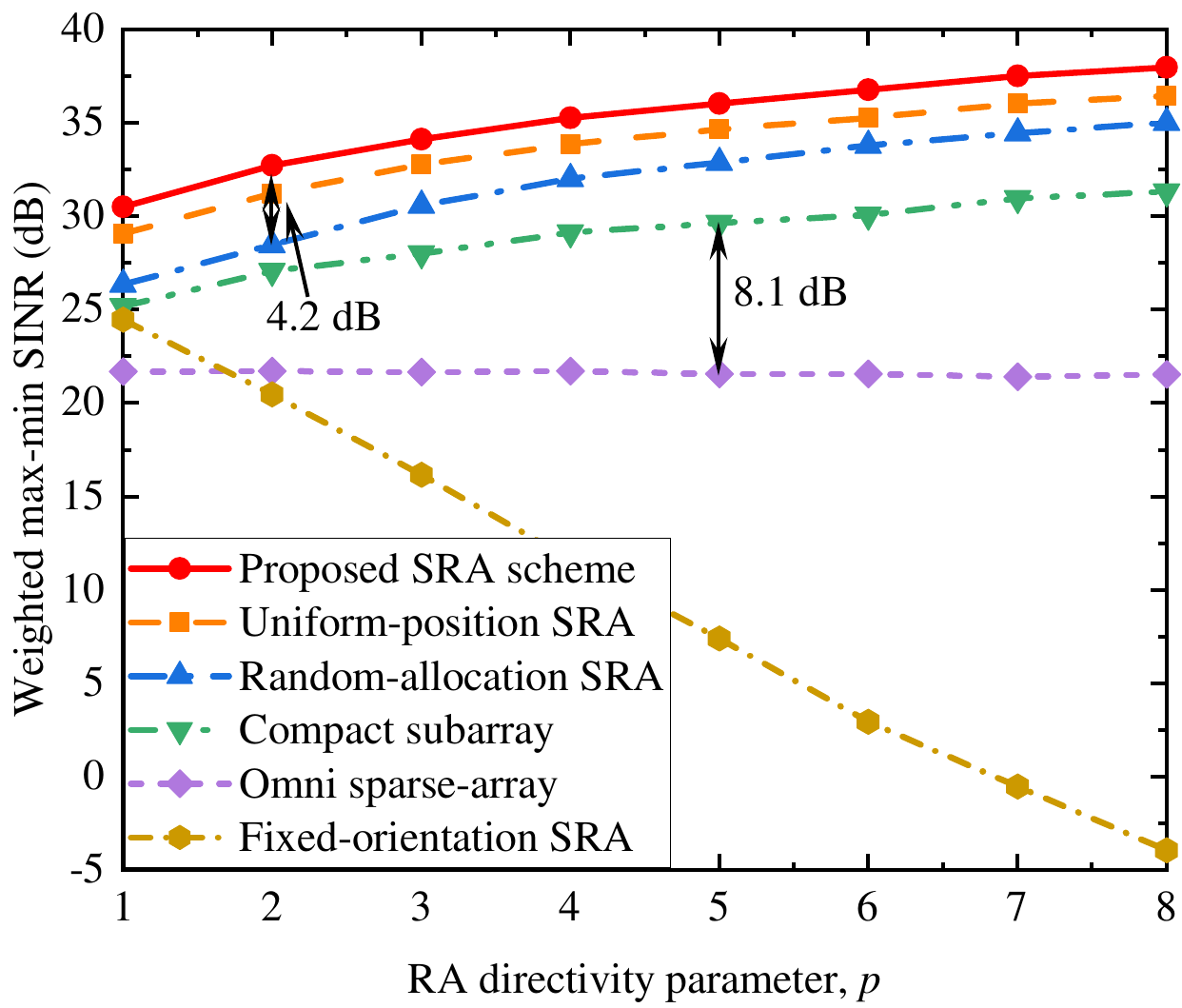}
		\caption{Weighted max-min SINR versus RA directivity
			parameter $p$.}
		\label{FigDire}
		\vspace{-6mm}
	\end{figure}
	
		\begin{figure}[t]
		\centering
		\includegraphics[width=0.36 \textwidth]{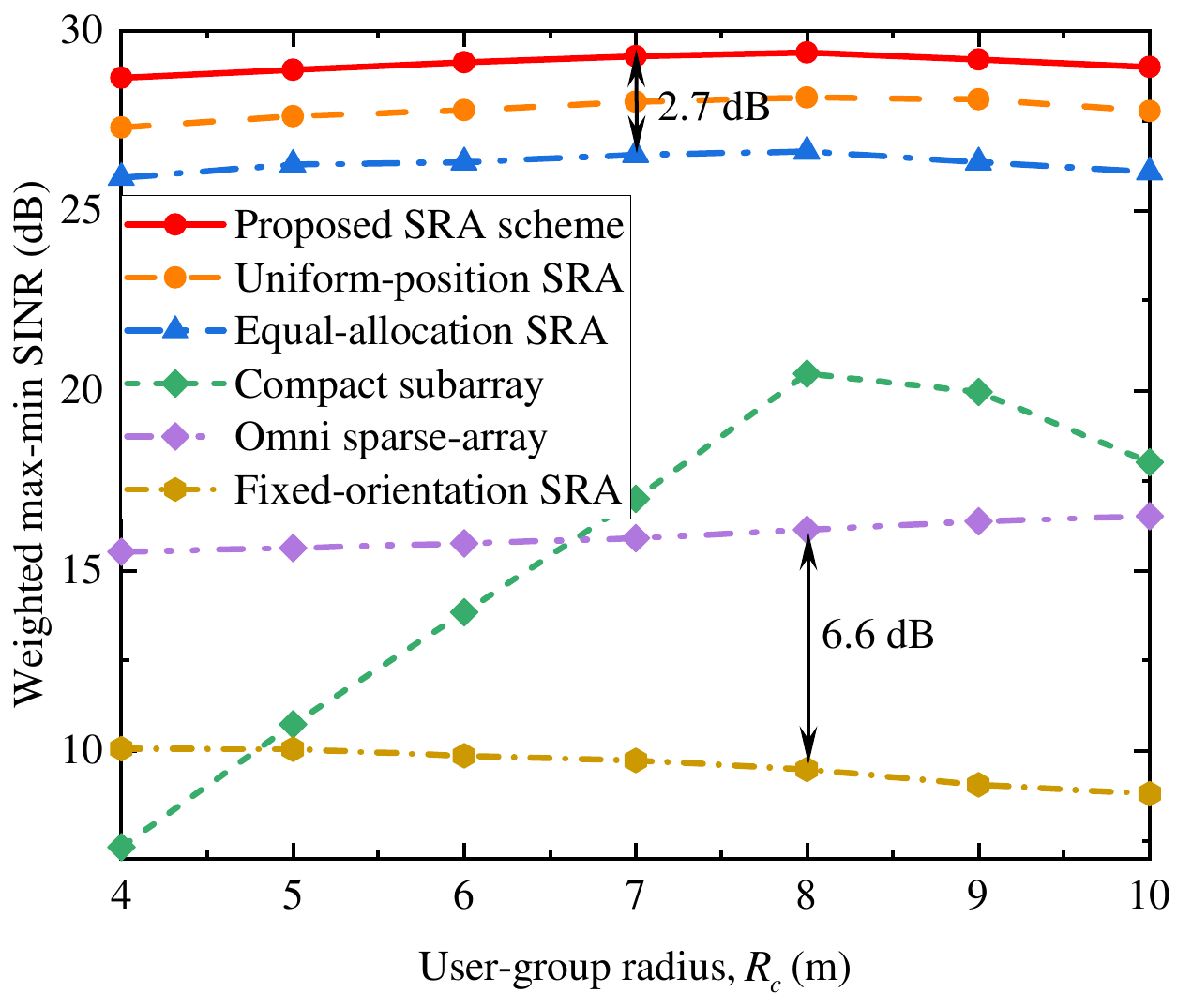}
		\caption{Weighted max-min SINR versus the user-group radius
			$R_c$.}
		\label{FigRc}
		\vspace{-6mm}
	\end{figure}
	
	\begin{figure*}[t]
		\centering
		
		\subfloat[Position-level sparse RA allocation map under different schemes]{
			\includegraphics[width=0.72\textwidth]{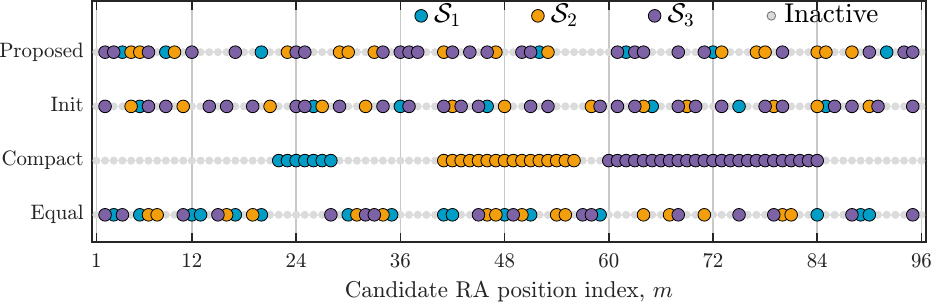}
			\label{Nallocation}
		}
		\hfill
		\subfloat[Illustration of group-wise effective aperture and channel correlation]{
			\includegraphics[width=0.72\textwidth]{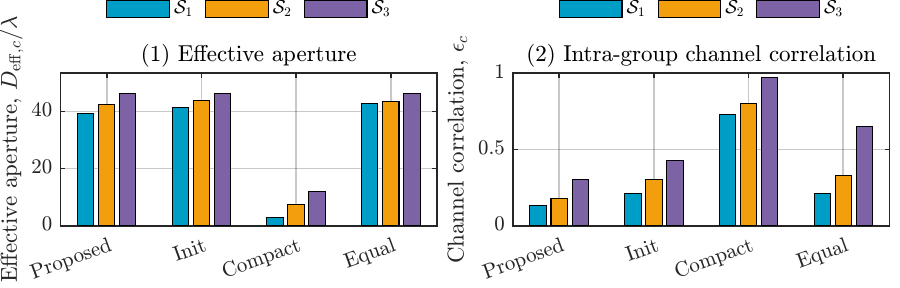}
			\label{NallocationCor}
		}
		\caption{Sparse RA allocation with $M=96$ and $N=48$: (a) position allocation; (b) group-wise effective aperture and channel correlation.}
		\label{NallocationAll}
		\vspace{-7mm}
	\end{figure*}

Fig.~\ref{FigDire} shows the weighted max-min SINR versus
the RA directivity parameter $p$. The performance of the
proposed SRA scheme improves with $p$, since a larger $p$ yields
a higher boresight gain and a narrower directional beam.
Under the group-center pointing rule, the enhanced antenna
directivity strengthens the desired links while suppressing
leakage toward non-target groups. At $p=2$, the proposed SRA
scheme outperforms the random-allocation SRA scheme by approximately
$4.2$-dB, demonstrating the importance of jointly optimizing
the group-wise antenna numbers and sparse positions even under
moderate RA directivity.
The compact subarray scheme also benefits from increasing
$p$, since the enhanced directivity improves inter-group
isolation. However, its performance remains inferior to the
proposed SRA scheme due to the limited effective aperture and
insufficient intra-group spatial separability.  At $p=5$, the compact subarray scheme
outperforms the omni sparse-array scheme by approximately $8.1$-dB,
confirming the effectiveness of antenna directivity in
mitigating inter-group interference. In contrast, the
fixed-orientation SRA scheme degrades as $p$ increases because its
boresights cannot adapt to the group locations, and the
increasingly narrow beams aggravate pointing mismatch and
off-boresight attenuation. These results demonstrate that the
performance gain of SRA comes from the joint exploitation of
large array aperture and antenna directivity.

	Fig.~\ref{FigRc} shows the weighted max-min SINR versus the 
	user-group radius $R_c$. Increasing $R_c$ initially improves 
	the angular and range separability among intra-group users by 
	providing larger spatial separation, but also increases the 
	off-boresight attenuation of edge users due to the group-center 
	pointing strategy. As a result, the proposed SRA scheme remains 
	relatively stable and decreases only slightly at large 
	$R_c$, since it adapts the sparse aperture distribution and 
	antenna numbers to balance spatial separability and directional 
	coverage. At $R_c=7$~m, it outperforms the equal-allocation SRA 
	scheme by approximately $2.7$-dB, demonstrating the benefit of 
	group-aware sparse aperture design. The compact subarray scheme 
	first benefits from the improved user separability and then 
	degrades as its limited aperture resolution and increased 
	off-boresight attenuation become dominant. In contrast, the 
	fixed-orientation SRA scheme deteriorates with $R_c$ because its 
	boresights cannot adapt to the changing user locations. At 
	$R_c=8$~m, the omni sparse-array scheme outperforms the fixed-orientation SRA
	scheme by approximately $6.6$-dB, highlighting the importance of 
	adaptive RA boresight steering for directional coverage.

		Fig.~\ref{NallocationAll} illustrates the allocation behavior with $M=96$ candidate positions and $N=48$ activated RAs. As shown in Fig.~\ref{NallocationAll}(a), the proposed SRA scheme distributes the activated RAs over the entire candidate aperture in a non-periodic manner, whereas the compact subarray scheme assigns contiguous subarrays to different groups. In Fig.~\ref{NallocationAll}(b), the compact subarray scheme has a much smaller effective aperture and suffers from higher intra-group channel correlation. Although the equal-allocation SRA scheme also spans a large aperture, its channel correlation remains higher than that of the proposed SRA scheme, showing that the proposed group-aware non-periodic allocation improves user separability beyond merely enlarging the aperture.

		\vspace{-4mm}
	\subsection{Large-Aperture ZF--MRT Beam-Direction Benchmark}
		\vspace{-1mm}
	To verify the group-wise ZF--MRT beam-direction approximation in Section~\ref{Performance Analysis}, we vary the candidate aperture size $M$ while fixing $\Delta=\lambda/2$ and $N$. 
	For each $M$, all beamforming schemes use the same antenna allocation for a fair comparison.
	
	For a fixed allocation, the group-wise ZF matrix is
	${\bf Q}_c^{\rm ZF}
	=
	\widetilde{\bf H}_{c,c}^H
	\left(
	\widetilde{\bf H}_{c,c}
	\widetilde{\bf H}_{c,c}^H
	\right)^{-1}$,
	where $\widetilde{\bf H}_{c,c}
=
[\widetilde{\bf h}_{c,c,1},\ldots,
\widetilde{\bf h}_{c,c,K_c}]^H$. 
	Let $\mathbf q_{c,k}^{\rm ZF}$ denote its $k$-th column. The normalized group-wise ZF and MRT directions are respectively given by
	\begin{equation}
		\label{MRT}
		\mathbf n_{c,k}^{\rm ZF}
		=
		{\mathbf q_{c,k}^{\rm ZF}}/
		{\Vert\mathbf q_{c,k}^{\rm ZF}\Vert_2},
		\mathbf n_{c,k}^{\rm MRT}
		=
		{\widetilde{\mathbf h}_{c,c,k}}/
		{\Vert\widetilde{\mathbf h}_{c,c,k}\Vert_2}.
		\end{equation}
		For comparison, we consider group-wise MRT with equal power allocation, referred to as MRT-EP. It adopts the same MRT beam directions ${\mathbf n_{c,k}^{\rm MRT}}$ in \eqref{MRT}, while assigning equal transmit power to all users, i.e.,
		$p_{c,k}^{\rm EP}
			=
			\frac{P_{\max}}{K}$.
	For the optimized-power ZF and MRT benchmarks, the transmit powers are obtained by bisection over the common weighted-SINR target, where each fixed-target feasibility test is a linear program. 
	Their phase-invariant beam-direction distance is defined as
\begin{align}
	d^{\rm BD}_{c,k}
	\!\!=\!\!\!
	\min_{\varphi\in[0,2\pi)} \!\!
	\left\|
	{\mathbf n}^{\rm ZF}_{c,k}
	\!-\!
	e^{j\varphi}{\mathbf n}^{\rm MRT}_{c,k}
	\right\|_2 
	\!\!=\!\!
	\sqrt{
		2  (
		1-
		|
		({\mathbf n}^{\rm ZF}_{c,k})^H 
		{\mathbf n}^{\rm MRT}_{c,k}
		|
		)
	}. \nonumber 
\end{align}
Accordingly, the worst-case beam direction distances are given by $d_{\rm ZF-MRT}=
\max_{c\in\mathcal C,\,k\in\mathcal K_c}
d^{\rm BD}_{c,k}$.

	\begin{figure}[t]
		\centering
		\includegraphics[width=0.4 \textwidth]{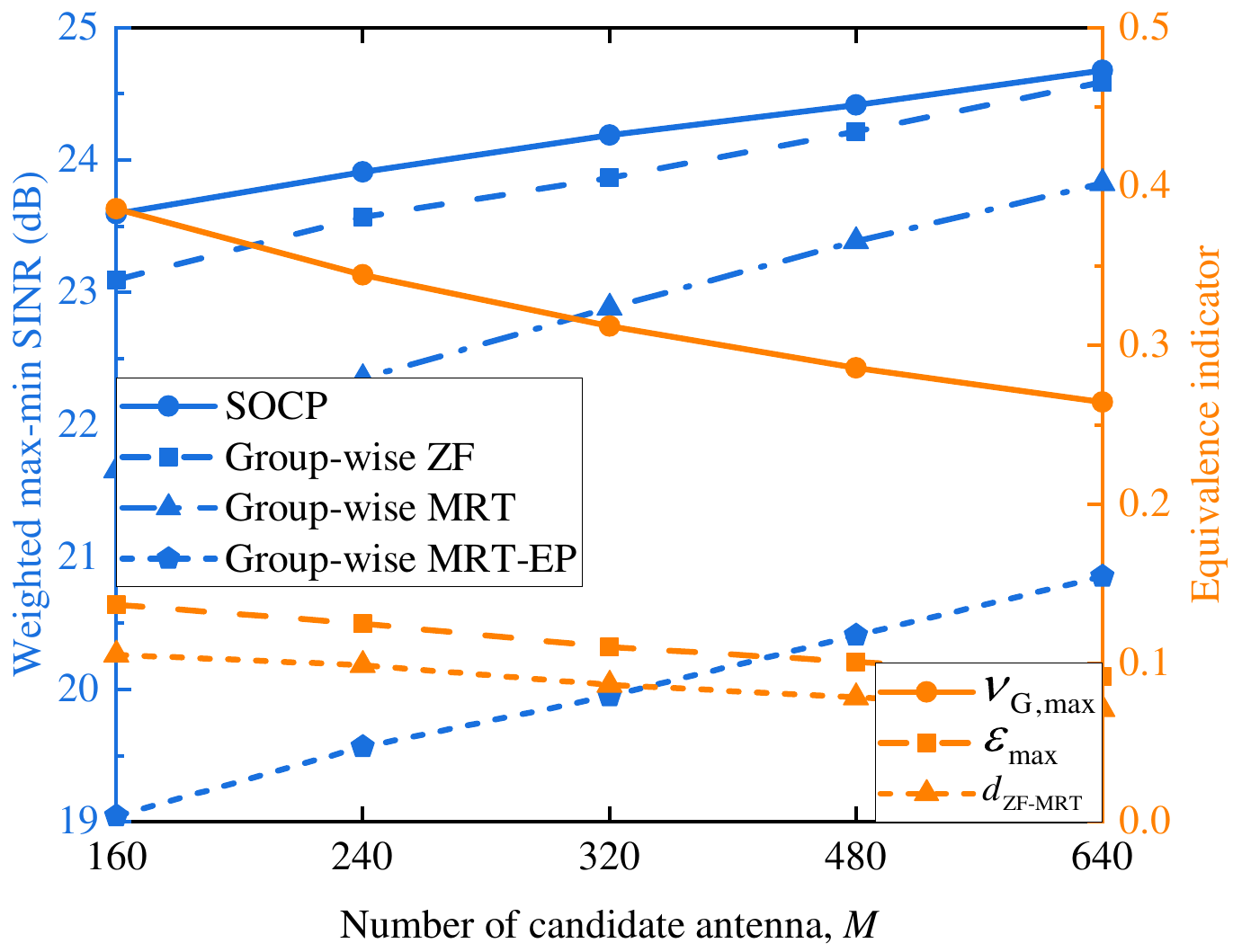} 
		\caption{Weighted max-min SINR and equivalence indicators
			versus the number of candidate RA positions $M$.}
		\label{FigLocalMRTZF}
		\vspace{-7mm}
	\end{figure}

Using the normalized intra-group channel matrix $\bar{\mathbf H}_c$ defined in Section~\ref{III-A}, we form the Gram matrix $\mathbf G_c\triangleq\bar{\mathbf H}_c\bar{\mathbf H}_c^{H}$ to characterize the joint intra-group channel conditioning. We further define
\begin{align}
	\nu_{\mathrm G,c}
	\triangleq
	1-
	\frac{\lambda{\min}(\mathbf G_c)}
	{\lambda_{\max}(\mathbf G_c)},
	\nu_{\mathrm G,\max}
	\triangleq
	\max_{c\in\mathcal C:K_c\geq2}
	\nu_{\mathrm G,c}.
\end{align}
A smaller $\nu_{\mathrm G,\max}$ indicates better-conditioned intra-group channels, with $\nu_{\mathrm G,\max}=0$ corresponding to orthogonality. The approximation scale in Lemma~\ref{pro1} is defined as
$\varepsilon_{\max}
\triangleq
\max\nolimits_{c\in\mathcal C:K_c\geq2}
\sqrt{K_c-1}\varepsilon_c$.
Fig.~\ref{FigLocalMRTZF} verifies the ZF--MRT beam-direction
approximation induced by the enlarged effective aperture under
$P_{\max}=20$ dBm. As the candidate aperture increases, the
improved spatial separability reduces intra-group channel
correlation and enhances the performance of all schemes. The
SOCP beamformer achieves the best performance, while group-wise
ZF closely approaches it by suppressing intra-group interference.
Group-wise MRT also benefits from the enlarged aperture but
remains inferior due to residual multiuser interference. Meanwhile,
$\nu_{\mathrm G,\max}$, $\varepsilon_{\max}$, and $d_{\rm ZF-MRT}$
all decrease with the aperture size, confirming improved channel
orthogonality and the ZF-to-MRT beam-direction convergence
predicted in Lemma~\ref{pro1}. The remaining gap between group-wise ZF and
MRT is attributed to finite-aperture and finite-SNR effects, where
MRT cannot completely eliminate intra-group interference.

	\begin{figure}[t]
		\centering
		\includegraphics[width=0.4 \textwidth]{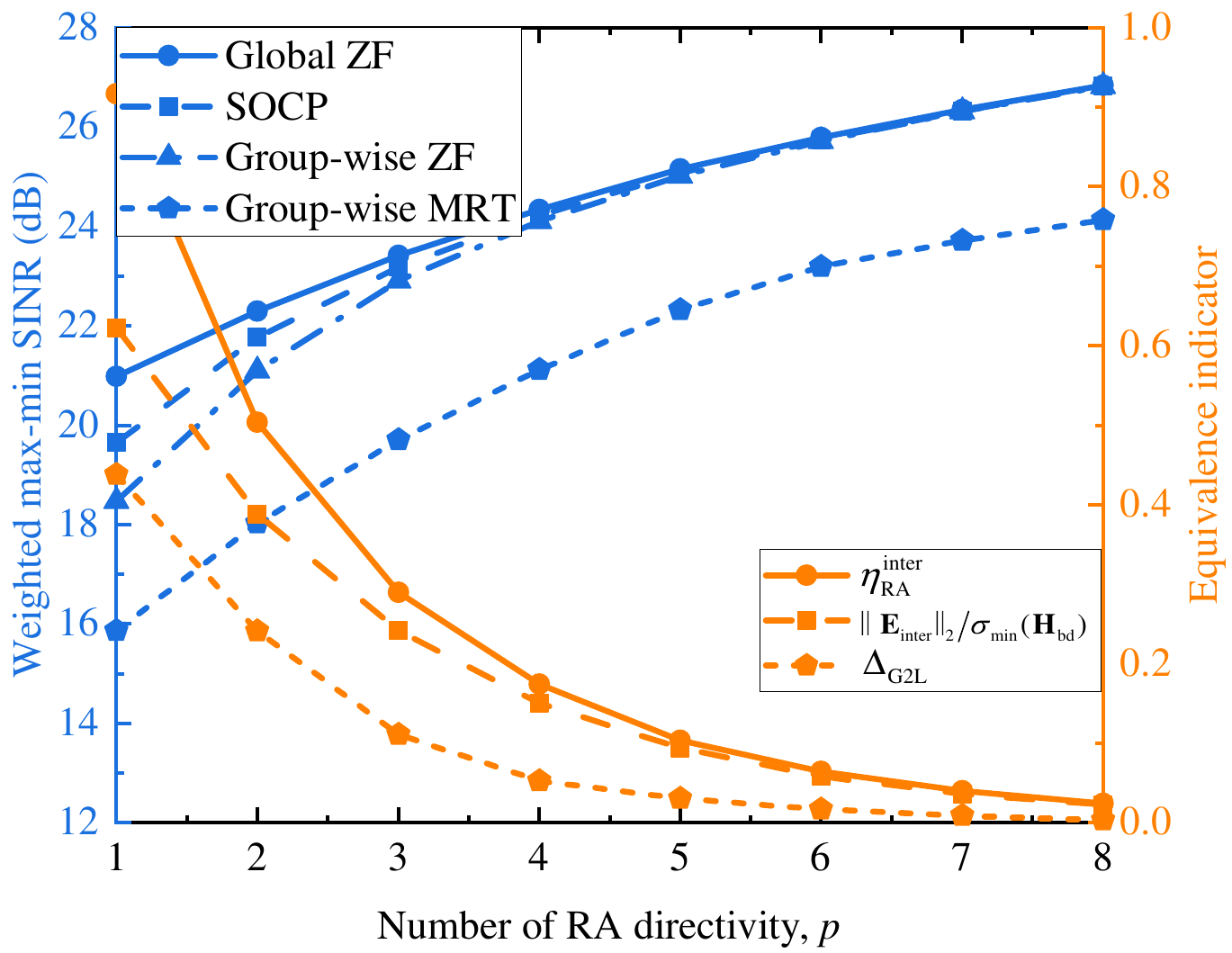}
		\caption{Weighted max-min SINR and equivalence indicators versus RA directivity $p$.} 
		\label{FigDirectivityBF}
		\vspace{-7mm}
	\end{figure}

Fig.~\ref{FigDirectivityBF} evaluates the impact of SRA antenna
directivity on the weighted max-min SINR and the global-to-group-wise
ZF approximation indicators versus the RA directivity parameter $p$. Let $\tau_{\rm GZF}$ and $\tau_{\rm LZF}$ denote the weighted max-min SINRs achieved by global ZF and group-wise ZF, respectively, under optimized power allocation. Their normalized performance gap, evaluated in the linear SINR domain, is defined as $\Delta_{\rm G2L}
	\triangleq
	\frac{
		\left|
		\tau_{\rm GZF}-\tau_{\rm LZF}
		\right|
	}{
		\tau_{\rm GZF}
	}$.
The global ZF precoder is constructed from the global reduced channel matrix $\mathbf H_{\rm g}$, whereas group-wise ZF is independently designed from the diagonal channel components associated with the group-specific subarrays. The indicators $\eta_{\rm RA}^{\rm inter}$, $\Vert\mathbf E_{\rm inter}\Vert_2/\sigma_{\min}(\mathbf H_{\rm bd})$, and $\Delta_{\rm G2L}$ quantify the worst-case RA leakage, the relative strength of the off-diagonal channel components, and the normalized SINR gap, respectively. 
 With the SRA group-center boresight control, increasing $p$ strengthens directional filtering, thereby suppressing inter-group leakage and improving the achievable
SINR. Meanwhile, all three indicators decrease, confirming the SRA directional analysis in Section~\ref{III-B} that the global channel approaches a block-diagonal structure and global ZF approaches group-wise ZF in the low-leakage regime.

		\vspace{-2mm}
	\section{Conclusion} \label{Conclusion}
		\vspace{-1mm}
		In this paper, we proposed an SRA
		architecture for multiuser communications,
		which jointly reconfigures the effective array aperture and antenna directivity through sparse activation and RA boresight adjustment.
 Specifically, for heterogeneous multi-group users,
		a weighted max-min SINR optimization problem was formulated by
		jointly designing the sparse-aperture allocation, RA orientations,
		and transmit beamforming. Then, we characterized the spatial operating principles of SRA, revealing that sparse-aperture configuration
		improves intra-group spatial separability, while antenna directivity
		suppresses inter-group coupling, thereby enabling simplified
		group-wise beamforming structures and providing physical guidance
		for sparse-aperture design. Based on these insights, a low-complexity
		AO framework was developed, where RA orientations are determined
		through a projected group-center rule and sparse-aperture allocation
		and transmit beamforming are iteratively optimized. 
		Numerical results demonstrated that the proposed SRA design achieves significant performance improvements over conventional sparse-array schemes by jointly exploiting large-aperture spatial DoFs and RA-enabled antenna directivity.

	\appendices
	\refstepcounter{section}
		\vspace{-5mm}
			\section*{APPENDIX~\Alph{section}: Proof of Lemma \ref{pro1}}
			\label{AppendixA}
			\vspace{-2mm}
		For user $k$, let $\bar{\mathbf H}_{c,-k}$ be obtained by
		removing the $k$-th row of $\bar{\mathbf H}_c$, and define
	$\mathbf a_{c,k}
			\triangleq
			\bar{\mathbf H}_{c,-k}\bar{\mathbf h}_{c,k}$.
		Under the condition $(K_c-1)\varepsilon_c<1$, the normalized
		Gram matrix
		$\bar{\mathbf H}_c\bar{\mathbf H}_c^H$ is positive definite.
		Indeed, its diagonal entries equal one, while the magnitudes
		of its off-diagonal entries are no larger than
		$\varepsilon_c$. By the Gershgorin circle theorem, we have
		\begin{equation}
			\lambda_{\min}
			\left(
			\bar{\mathbf H}_c\bar{\mathbf H}_c^H
			\right)
			\geq
			1-(K_c-1)\varepsilon_c
			>
			0.
			\label{eq:app_gram_positive}
		\end{equation}
		Therefore, the group-wise ZF direction is well defined.
		
		The group-wise ZF direction for user $k$ is parallel to the
		component of $\bar{\mathbf h}_{c,k}$ orthogonal to the
		subspace spanned by the other users' normalized channels.
		Define the corresponding orthogonal projection matrix as
		$\mathbf P_{c,-k}
			\triangleq
			\bar{\mathbf H}_{c,-k}^{H}
			(
			\bar{\mathbf H}_{c,-k}
			\bar{\mathbf H}_{c,-k}^{H}
			)^{-1}
			\bar{\mathbf H}_{c,-k}$,
		and let $\eta_{c,k}
			\triangleq
			\left\|
			\mathbf P_{c,-k}\bar{\mathbf h}_{c,k}
			\right\|_2$.
		Since $\bar{\mathbf H}_c\bar{\mathbf H}_c^H$ is positive
		definite, $\bar{\mathbf h}_{c,k}$ does not belong to the
		subspace spanned by the other users' channels. Hence,
		$0\leq\eta_{c,k}<1$.
		The normalized group-wise ZF direction is parallel to
		$(\mathbf I-\mathbf P_{c,-k})\bar{\mathbf h}_{c,k}$.
		Moreover,
		$\|
			(\mathbf I-\mathbf P_{c,-k})
			\bar{\mathbf h}_{c,k}
			\|_2
			=
			\sqrt{1-\eta_{c,k}^2}$.
		Therefore, after phase alignment, the ZF--MRT
		beam-direction distance satisfies
		\begin{align}
			\left(
			d_{c,k}^{\rm ZF-MRT}
			\right)^2
			&=
			2(
			1-\sqrt{1-\eta_{c,k}^2}
			)\leq
			2\eta_{c,k}^2,
			\label{eq:app_distance_eta}
		\end{align}
		where the inequality follows from
		$1-\sqrt{1-x}\leq x$ for $x\in[0,1]$. Thus, we have
		$d_{c,k}^{\rm ZF-MRT}
			\leq
			\sqrt{2}\eta_{c,k}$.
		It remains to bound $\eta_{c,k}$. We define
		$\mathbf R_{c,-k}
			\triangleq
			\bar{\mathbf H}_{c,-k}
			\bar{\mathbf H}_{c,-k}^{H}$.
		We have $\eta_{c,k}^2
			=
			\mathbf a_{c,k}^{H}
			\mathbf R_{c,-k}^{-1}
			\mathbf a_{c,k}$.
		The matrix $\mathbf R_{c,-k}$ has unit diagonal entries,
		while the magnitudes of its off-diagonal entries are no larger
		than $\varepsilon_c$. Hence, the Gershgorin circle theorem
		gives $\lambda_{\min}
			\left(
			\mathbf R_{c,-k}
			\right)
		\!	\geq \!
			1 \!-\! (K_c \! -\! 2)\varepsilon_c$.
		Then, we obtain
		\begin{equation}
			\eta_{c,k}
			\leq
			\frac{
				\|\mathbf a_{c,k}\|_2
			}{
				\sqrt{1-(K_c-2)\varepsilon_c}
			}.
			\label{eq:app_eta_final}
		\end{equation}
		Furthermore, we have
		\begin{align}
			\|\mathbf a_{c,k}\|_2^2
			&=
			\sum\nolimits_{i\neq k}
			\left|
			\bar{\mathbf h}_{c,i}^{H}
			\bar{\mathbf h}_{c,k}
			\right|^2\leq
			(K_c-1)\varepsilon_c^2.
			\label{eq:app_a_final}
		\end{align}
		Finally, we obtain
		\begin{equation}
			d_{c,k}^{\rm ZF-MRT}
			\leq
			\frac{
				\sqrt{2(K_c-1)}\varepsilon_c
			}{
				\sqrt{1-(K_c-2)\varepsilon_c}
			},
		\end{equation}
		which proves \eqref{eq:zf_mrt_bound}.
		For fixed $K_c$, the denominator approaches one as
		$\varepsilon_c\to0$. Therefore, we have
		\begin{equation}
			d_{c,k}^{\rm ZF-MRT}
			=
			\mathcal O
			(
			\sqrt{K_c-1}\varepsilon_c
			),
		\end{equation}
		which establishes the asymptotic statement in
		Lemma~\ref{pro1}. This completes the proof.

		\refstepcounter{section}
		\vspace{-4mm}
		\section*{APPENDIX~\Alph{section}: Proof of Lemma \ref{pro2}}
	\label{app2}
	\vspace{-2mm}
	Since each diagonal block $\mathbf H_{c,c}$ has full row rank,
	the block-diagonal matrix $\mathbf H_{\rm bd}$ also has full row rank.
	Since $\mathbf H_g
		=
		\mathbf H_{\rm bd}
		+
		\mathbf E_{\rm inter}$,
under the condition
	\begin{equation}
		\|\mathbf E_{\rm inter}\|_2
		\leq
		\rho\sigma_{\min}(\mathbf H_{\rm bd})=\rho	s_{\rm bd},
		 0<\rho<1,
	\end{equation}
	Weyl's inequality gives
	\begin{align}
		\sigma_{\min}(\mathbf H_g)
		\!\geq \!
		\sigma_{\min}(\mathbf H_{\rm bd})
		\!-\!
		\|\mathbf E_{\rm inter}\|_2 
		\!\geq \!
		( 1 \!-\! \rho)	s_{\rm bd}
		\!>\!0. \nonumber 
	\end{align}
	Therefore, $\mathbf H_g$ has full row rank, and the global ZF
	precoder is well defined.
Then, we define
	\begin{align}
		\!\!\mathbf A
		\!\!=\!\!
		\mathbf H_{\rm bd}\mathbf H_{\rm bd}^{H},
		\boldsymbol{\Delta}_{\rm inter}
		\!=\!
		\mathbf H_{\rm bd}\mathbf E_{\rm inter}^{H}
		\!+\!
		\mathbf E_{\rm inter}\mathbf H_{\rm bd}^{H}
		\!+\!
		\mathbf E_{\rm inter}\mathbf E_{\rm inter}^{H}. \nonumber
	\end{align}
	Then, we get $\mathbf H_g\mathbf H_g^{H}=
	\mathbf A
	+
	\boldsymbol{\Delta}_{\rm inter}$ and
	\begin{align}
		\|\boldsymbol{\Delta}_{\rm inter}\|_2
		&\leq
		2\|\mathbf H_{\rm bd}\|_2\|\mathbf E_{\rm inter}\|_2
		+
		\|\mathbf E_{\rm inter}\|_2^2 \nonumber \\
		&\leq
		\left(
		2\|\mathbf H_{\rm bd}\|_2
		+
		\rho	s_{\rm bd}
		\right)
		\|\mathbf E_{\rm inter}\|_2.
	\end{align}
	Furthermore, we obtain
	\begin{align}
		\!\left\|
		(\mathbf H_g\mathbf H_g^{H})^{-1}
		\right\|_2
		\!\leq\!
		\frac{1}{
			(1\!-\! \rho)^2
				s_{\rm bd}^2
		},
		\|\mathbf A^{-1}\|_2
		\!=\!
		\frac{1}{	s_{\rm bd}^2}. \nonumber 
	\end{align}
	Using the matrix inverse perturbation identity
$(\mathbf A+\boldsymbol{\Delta}_{\rm inter})^{-1}
		-
		\mathbf A^{-1}
		=
		-(\mathbf A+\boldsymbol{\Delta}_{\rm inter})^{-1}
		\boldsymbol{\Delta}_{\rm inter}
		\mathbf A^{-1}$,
	we obtain
	\begin{align}
		\mathbf W_g^{\rm ZF}
	\!	\!-\!\!
		\mathbf W_{\rm gw}^{\rm ZF}
	\!	=\!
		\mathbf E_{\rm inter}^{H}
		(\mathbf H_g\mathbf H_g^{H})^{-1}
		\!\!-\!\!
		\mathbf H_{\rm bd}^{H}
		(\mathbf H_g\mathbf H_g^{H})^{-1}
		\boldsymbol{\Delta}_{\rm inter}
		\mathbf A^{-1}, \nonumber 
	\end{align}
	where $	\mathbf W_{\rm gw}^{\rm ZF}
		=
		\mathbf H_{\rm bd}^{H}
		(\mathbf H_{\rm bd}\mathbf H_{\rm bd}^{H})^{-1}$.
	Therefore, we have 
		\begin{align}
		&
		\left\|
		\mathbf W_g^{\rm ZF}
		-
		\mathbf W_{\rm gw}^{\rm ZF}
		\right\|_2\leq
		\frac{\|\mathbf E_{\rm inter}\|_2}
		{(1-\rho)^2s_{\rm bd}^2}
		+
		\frac{
			\|\mathbf H_{\rm bd}\|_2
			\|\bm{\Delta}_{\rm inter}\|_2
		}{
			(1-\rho)^2s_{\rm bd}^4
		}
		\nonumber\\
		&\leq
		\left[
		\frac{1}
		{(1-\rho)^2s_{\rm bd}^2}
		+
		\frac{
			\|\mathbf H_{\rm bd}\|_2
			\left(
			2\|\mathbf H_{\rm bd}\|_2
			+
			\rho s_{\rm bd}
			\right)
		}{
			(1-\rho)^2s_{\rm bd}^4
		}
		\right]
		\|\mathbf E_{\rm inter}\|_2. \nonumber
	\end{align}
		Finally, applying
	$\|\mathbf H_{\rm bd}\|_2=\kappa_{\rm bd}s_{\rm bd}$,
we get
	\begin{align}
		&
		\frac{1}
		{(1\!-\! \rho)^2s_{\rm bd}^2}
		\!\!+\!\!
		\frac{
			\|\mathbf H_{\rm bd}\|_2
			\left(
			2\|\mathbf H_{\rm bd}\|_2
			\!+\!
			\rho s_{\rm bd}
			\right)
		}{
			(1-\rho)^2s_{\rm bd}^4
		}\!=\!
		\frac{
			1
		\!	+\!
			2\kappa_{\rm bd}^2
			\!+\!
			\rho\kappa_{\rm bd}
		}{
			(1\!-\! \rho)^2s_{\rm bd}^2
		}. \nonumber
	\end{align}
	Therefore, we obtain
	\begin{equation}
		\left\|
		\mathbf W_g^{\rm ZF}
		-
		\mathbf W_{\rm gw}^{\rm ZF}
		\right\|_2
		\leq
		\frac{
			1
			+
			2\kappa_{\rm bd}^2
			+
			\rho\kappa_{\rm bd}
		}{
			(1-\rho)^2s_{\rm bd}^2
		}
		\|\mathbf E_{\rm inter}\|_2,
	\end{equation}
	which proves \eqref{eq:global_groupwise_zf_bound}.
	This completes the proof.

		\vspace{-3mm}
	\bibliographystyle{IEEEtran}
	\vspace{-1mm}
	\begin{spacing}{0.92}
		\bibliography{thesis}
	\end{spacing}
	
\end{document}